\documentclass[draftcls, onecolumn, 12pt]{IEEEtran}
\usepackage{amsmath,graphicx, amsthm}

\usepackage{tikz}
\usetikzlibrary{arrows,automata}

\usepackage{dsfont}
\usepackage{algorithmicx}
\usepackage{algpseudocode}
\usepackage{caption}
\usepackage{subcaption}
\usepackage{colortbl} 

\newtheorem{theorem}{Theorem}[section]
\newtheorem{lemma}{Lemma}[section]

\begin{document}

\title{Contact Process with Exogenous Infection and the Scaled SIS Process}

\author{\IEEEauthorblockN{June Zhang}\\
\IEEEauthorblockA{Carnegie Mellon University, Dept. of Electrical and Computer Engineering,\\
Pittsburgh, PA, USA\\
junez@andrew.cmu.edu\\}
\and
\IEEEauthorblockN{Jos\'{e} M.F.~Moura}\\
\IEEEauthorblockA{Carnegie Mellon University, Dept. of Electrical and Computer Engineering\\
Pittsburgh, PA, USA\\
moura@ece.cmu.edu}}
\thanks{This work is partially supported by AFOSR grant FA95501010291, and by NSF grants CCF1011903 and CCF1018509.}

\maketitle

\begin{abstract}
Propagation of contagion in networks depends on the graph topology. This paper is concerned with studying the time-asymptotic behavior of the extended contact processes on static, undirected, finite-size networks. This is a contact process with nonzero exogenous infection rate (also known as the $\epsilon$-SIS, $\epsilon$ susceptible-infected-susceptible, model \cite{PhysRevE.86.016116}). The only known analytical characterization of the equilibrium distribution of this process is for complete networks. For large networks with arbitrary topology, it is infeasible to numerically solve for the equilibrium distribution since it requires solving the eigenvalue-eigenvector problem of a matrix that is exponential in $N$, the size of the network. We show that, for a certain range of the network process parameters, the equilibrium distribution of the extended contact process on \emph{arbitrary}, finite-size networks is well approximated by the equilibrium distribution of the scaled SIS process, which we derived in closed-form in prior work. We confirm this result with numerical simulations comparing the equilibrium distribution of the extended contact process with that of a scaled SIS process. We use this approximation to decide, in \emph{polynomial-time}, which agents and network substructures are more susceptible to infection by the extended contact process.
\end{abstract}
\begin{keywords}
scaled SIS process, contact process, $\epsilon$-SIS, epidemics on networks, networks, subgraph density 
\end{keywords}

\section{Introduction}

The contact process \cite{liggett1999stochastic} and its extension, the $\epsilon$-SIS (susceptible-infected-susceptible) model \cite{PhysRevE.86.016116}, which we will refer to as the extended contact process in this paper, are models widely considered for describing the propagation dynamics of failures or epidemics in complex networks; the network describes and constrains the interactions and interdependencies between multiple agents/components in the system \cite{barrat2008dynamical, keeling2005networks, porter2014dynamical, pellis2014eight}. We call models of such phenomena \emph{network processes}. Network processes extend traditional dynamical processes since the network substrate itself is a determinant of the observed dynamics. Except in special cases, such as when the network is a complete graph or the network is comprised of isolated nodes, it is a challenge to analyze how network topology influences the process dynamics. In network science, developing analyzable models that quantify the impact of topology on the behavior of network processes remains an open question.

We are interested in understanding the role that topology plays on the time-asymptotic behaviors of network processes. For continuous-time Markov processes such as the contact and the extended contacted processes, the time-asymptotic behavior is characterized by their equilibrium distribution (i.e., the limiting distribution of the process). The equilibrium distribution of the contact process on networks with $N$ agents ($N < \infty$) is trivial, because it assumes healing and infections are only due to contagion from infected neighbors \cite{liggett1999stochastic}. The contact process has an absorbing state; its equilibrium distribution is zero everywhere, except one at the absorbing state. References \cite{Ganesh, Draief:2006:TVS:1190095.1190160} looked at the effect of topology on the time it takes for the process to reach steady state.

Due to the inclusion of a non-zero exogenous infection rate, the extended contact process does not have a trivial equilibrium distribution. In general, to compute the equilibrium distribution for this network process requires solving for the left eigenvector corresponding to the zero eigenvalue of the $2^N \times 2^N$ transition rate matrix, $\mathbf{Q}_e$; this is an infeasible computation problem for networks with more than a few agents. For large-scale networks, researchers usually approximate large-scale networks by infinite-size networks using the mean-field approximation \cite{PhysRevE.86.016116}. We take a different approach and show that, for a subclass of extended contact processes, their equilibrium distribution can be approximated by that of the scaled SIS process, for which we found the closed-form equilibrium distribution of the process on \emph{arbitrary}, undirected, finite-size network topology \cite{JZhangJournal, JZhangJournal2}. Unlike the extended contact process, which assumes that the infection rate of a healthy agent is linearly dependent on its number of infected neighbors, the scaled SIS process assumes an exponential dependence. 

We use this analytical characterization of the equilibrium distribution of the scaled SIS process as an approximation to the equilibrium distribution of the extended contact process. The paper shows that this approximation is appropriate for a range of endogenous infection rates of the extended contact process; it shows this range depends on the maximum degree of the underlying network. With numerical simulations, we confirm that, within this parameter range, the deviation between the true equilibrium distribution of the extended contact process and the approximation is very small (on the order of $10^{-5}$). Further, we observe from experiments that for certain network topologies, the approximation remains good even as the infection rate deviates from the established range.

We use the equilibrium distribution to address important questions regarding the extended contact process like deciding which agents and network structures are more susceptible to infection by solving for the most-probable configuration, configuration with maximum equilibrium probability. When the approximation holds, the most-probable configuration of the extended contact process is the same as the scaled SIS process, which we proved we can find in polynomial-time in \cite{JZhangJournal2}. 

In Section~\ref{sec:contact}, we review the contact process and the extended contact process. We review the scaled SIS process and compare and contrast it with the extended contact process in Section~\ref{sec:scaledSIS}. In Section~\ref{sec:timeasy}, we show a bound on the endogenous infection rate such that the equilibrium distribution of the extended contact process is well approximated by that of the scaled SIS process. We compare the true equilibrium distribution of the extended contact process with its approximation for six different 16-node networks using the total variation distance in Section~\ref{sec:numerical}. We discuss the most-probable configuration of the extended contact process and the approximate distribution in Section~\ref{sec:vulnerablesub}. Section~\ref{sec:conclusion} concludes the paper.

\section{Contact Process}\label{sec:contact}

The contact process models the spread of infection in a network \cite{liggett1999stochastic}. It is a binary state, irreducible, continuous-time Markov process on a static, simple, connected, undirected network $G$. See \cite{norris1998markov, Kelly} for review of continuous-time Markov processes, \cite{west2001introduction, algraph} for review of graph theory. Each node in the network is an agent in the population. Each node can be in one of two states, $\{0,1\}$, representing, for example, healthy or infected state. For a system with $N$ nodes, the microscopic network configuration is
\[
\mathbf{x} = [x_1, x_2, \ldots x_N]^T, \text{ where } x_i = \{0,1\}.
\]
As a result, there are $2^N$ possible configurations. 

The contact process models SIS (susceptible-infected-susceptible) epidemics on networks. There are two types of state transitions representing 1) healing of infected agents and 2) infection of susceptible agents.

\begin{enumerate}
\item
Consider the configuration
\[ 
\mathbf{x} = [x_1,x_2, \ldots, x_j = 1, x_k, \ldots x_N]^T.
\] 
Let $T^-_j\mathbf{x}$ be the configuration where the $j$th agent heals: 
\[
T^-_j\mathbf{x} =  [x_1,x_2, \ldots, x_j = 0, x_k, \ldots x_N]^T.
\] 
The contact process transitions from $\mathbf{x}$ to  $T^-_j\mathbf{x}$ in an exponentially distributed random amount of time with transition rate
\begin{equation}\label{eq:contacthealrate}
q(\mathbf{x}, T^-_j\mathbf{x}) = \mu.
\end{equation}
Parameter $\mu$ is the \emph{healing rate}. Without loss of generality, typically $\mu \equiv 1$. 

\item Consider the configuration 
\[
\mathbf{x} = [x_1,x_2, \ldots, x_j, x_k = 0, \ldots x_N]^T.
\]
Let $T^+_k\mathbf{x}$ be the configuration where the $k$th agent becomes infected:
\[
T^+_k\mathbf{x} = [x_1,x_2, \ldots, x_j, x_k = 1, \ldots x_N]^T.
\]
The contact process transitions from $\mathbf{x}$ to  $T^+_k\mathbf{x}$ in an exponentially distributed random amount of time with transition rate
\begin{equation}\label{eq:contactinfectrate}
q(\mathbf{x}, T^+_k\mathbf{x}) = \beta_e\sum_{i =1}^N x_iA_{ik},
\end{equation}
where $A = [A_{ik}]$ is the adjacency matrix of the underlying network. The parameter $\beta_e > 0$ is the \emph{endogenous infection rate}. The infection rate of the $k$th agent is assumed to be linearly dependent on its number of infected neighbors, $m = \sum_{i =1}^N x_iA_{ik}$.
\end{enumerate}

In the contact process, when all the agents in the network are healthy, the process dies out. The configuration where all the agents are healthy ($\mathbf{x}^0 = [0,0,\ldots, 0]^T$) is an absorbing state of the Markov process. For networks with $N$ agents, the contact process will eventually reach the configuration $\mathbf{x}^0$ and remain there indefinitely. Thus, the equilibrium distribution is trivial for contact processes on finite-size networks \cite{liggett1999stochastic}. 

%

\subsection{Extended Contact Process}
In the contact process, a healthy agent can only become infected through contagion from an infected neighbor. It may be the case that a healthy agent (or working component) may also become infected (or fail) due to an exogenous (i.e., outside of the network) source \textemdash the agent is infected spontaneously \cite{Augusto, PhysRevE.86.016116, JZhang}. For SIS epidemics, this is captured by a non-zero \emph{exogenous infection rate}, $\lambda$. The transition rate of the extended contact process from $\mathbf{x}$ to  $T^+_k\mathbf{x}$ is
\begin{equation}\label{eq:extendedcontactinfectrate}
q(\mathbf{x}, T^+_k\mathbf{x}) = \lambda + \beta_e\sum_{i =1}^N x_iA_{ik},
\end{equation}
where $A$ is the adjacency matrix of the underlying network. The healing rate remains the same as \eqref{eq:contacthealrate}. We call this modified model the \emph{extended contact process}, whereas \cite{PhysRevE.86.016116} referred to it as the $\epsilon$-SIS model. When agent $k$ has $0$ infected neighbors, the rate at which agent $k$ becomes infected is the exogenous infection rate. For a system where spontaneous infection is rare, the exogenous infection rate can be made arbitrarily small, but for the extended contact process, it has to remain greater than zero. 

The configuration where all the agents are healthy ($\mathbf{x}^0 = [0,0,\ldots, 0]^T$) is no longer an absorbing state in the Markov process since susceptible agents can spontaneously become infected.  As a result, the equilibrium distribution of the Markov process is no longer trivial. There is currently no known tractable analytical results regarding this equilibrium distribution for the extended contact process for \emph{arbitrary} network topologies; reference \cite{PhysRevE.86.016116} provided the exact equilibrium distribution only for the complete graph. 

The equilibrium distribution can be calculated numerically. However, this approach is infeasible for large networks. In the case of an irreducible, continuous-time Markov process, the equilibrium distribution, $\pi$, is the left eigenvector of the transition rate matrix, $\mathbf{Q}_e$, corresponding to the 0 eigenvalue. However, the transition rate matrix is a $2^N \times 2^N$ matrix, where $N$ is the size of the network. Solving for the equilibrium distribution of the extended contact process over a 200-node network with arbitrary topology means finding the eigenvector of a $2^{200} \times 2^{200}$ matrix; even taking into account sparsity, such computation is clearly infeasible. 

We will show in this paper that we can obtain an approximation to the equilibrium distribution over arbitrary network topologies for a subset of extended contact processes using the scaled SIS process.


\section{Scaled SIS Process}\label{sec:scaledSIS}
We introduced the scaled SIS process in \cite{JZhangJournal, JZhang2}. Like the contact process, it is a binary state, irreducible, continuous-time Markov process on static, simple, connected, undirected networks. The scaled SIS process assumes that an agent can be in one of 2 states, i.e., $\{0,1\}$, representing an healthy or infected state. The space of possible network configurations of the scaled SIS process is the same as that of the contact and extended contact process. 
The scaled SIS process also accounts for two types of state transitions representing 1) healing of infected agents and 2) infection of susceptible agents.
\begin{enumerate}
\item Consider the configuration
\[ 
\mathbf{x} = [x_1,x_2, \ldots, x_j = 1, x_k, \ldots x_N]^T.
\] 
Let $T^-_j\mathbf{x}$ be the configuration where the $j$th agent heals: 
\[
T^-_j\mathbf{x} =  [x_1,x_2, \ldots, x_j = 0, x_k, \ldots x_N]^T.
\] 
The scaled SIS process transitions from $\mathbf{x}$ to  $T^-_j\mathbf{x}$ in an exponentially distributed random amount of time with transition rate
\begin{equation}\label{eq:scaledhealrate}
q(\mathbf{x}, T^-_j\mathbf{x}) = \mu.
\end{equation}
The healing rate of the scaled SIS process is the same as the healing rate of the contact process.

\item Consider the configuration 
\[
\mathbf{x} = [x_1,x_2, \ldots, x_j, x_k = 0, \ldots x_N]^T.
\]
Let $T^+_k\mathbf{x}$ be the configuration where the $k$th agent becomes infected:
\[
T^+_k\mathbf{x} = [x_1,x_2, \ldots, x_j, x_k = 1, \ldots x_N]^T.
\]
The scaled SIS process transitions from $\mathbf{x}$ to  $T^+_k\mathbf{x}$ in an exponentially distributed random amount of time with transition rate
\begin{equation}\label{eq:scaledinfectrate}
q(\mathbf{x}, T^+_k\mathbf{x}) = \lambda\beta_s^{\sum_{i =1}^N x_iA_{ik}},
\end{equation}
where $A$ is the adjacency matrix of the underlying network. The parameter $\beta_s > 0$ is the endogenous infection rate. Unlike the infection rate of the extended contact process \eqref{eq:extendedcontactinfectrate}, the infection rate of the $k$th agent is assumed to be \emph{exponentially} dependent on the number of infected neighbors, $m = \sum_{i =1}^N x_iA_{ik}$. When the number of infected neighbors of agent $k$ is 0, the infection rate, like for the extended contact process, reduces to the exogenous infection rate $\lambda$. 
\end{enumerate}

We proved in \cite{JZhangJournal, JZhang2}, that, for the scaled SIS process, the resulting continuous-time Markov process is a \emph{reversible} Markov process; a reversible Markov process is a stochastic process that is statistically the same forward in time as it is in reverse \cite{Kelly}. The equilibrium distribution of a reversible Markov process is unique. The equilibrium distribution of the scaled SIS process over any undirected network topology described by the adjacency matrix, $A$, is found in \cite{JZhangJournal, JZhang2} to be:
\begin{equation} \label{eq:equilibriumdistribution}
\pi({\bf x}) =\frac{1}{Z}\left( \frac{\lambda}{\mu}\right)^{1^T{\mathbf x}}  \beta_s^{\frac{{\mathbf x}^TA{\mathbf x}}{2}  }, \quad  \mathbf{x} \in \mathcal{X}
\end{equation}
where $1$ is the vector of all 1's, $Z$ is the partition function, and $\mathcal{X}$ is the space of $2^N$ possible configurations \cite{JZhangJournal, JZhang2}. The equilibrium probability of a configuration $\mathbf{x}$ depends on the number of infected agents, $1^T{\mathbf x}$, and on the number of edges where both end nodes are infected, $\frac{{\mathbf x}^TA{\mathbf x}}{2}$.


\subsection{Scaled SIS Process vs. Extended Contact Process}\label{sec:scaledcontact}

The infection rate of a susceptible agent in both the extended contact process and the scaled SIS process depends on its number of infected neighbors. The two models make  different assumptions regarding the underlying mechanism of the contagion process:

\begin{figure}[ht]
        \centering
                \includegraphics[width=0.7\textwidth]{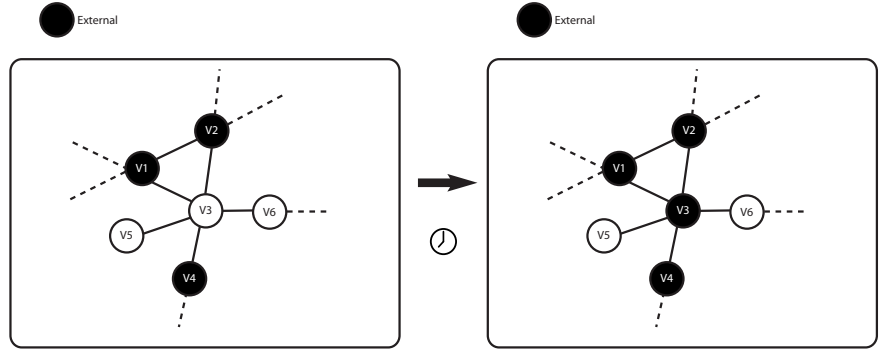}
	\caption{Transition from Configuration $\mathbf{x}$ to Configuration $T_3^+\mathbf{x}$}\label{fig:contact}
\end{figure}

\begin{description}
\item[Extended Contact Process] \hfill \\
The extended contact process is parameterized by the exogenous infection rate, $\lambda$, the healing rate, $\mu$, and the endogenous infection rate $\beta_e$. Consider the scenario in Figure~\ref{fig:contact}. Let $T_3$ be the random amount of time it takes for agent $V_3$ to become infected. Each infected neighbors of agent $V_3$ (i.e., $V_1, V_2, V_4$) and the exogenous (i.e., external) source may infect $V_3$ in an exponentially distributed amount of time $T_3^i \sim \exp(\beta_e), i = 1,2,4,$ and $T_3^e \sim \exp(\lambda)$, respectively. Therefore, $T_3 = \min\{T_3^1, T_3^2, T_3^{4}, T_3^{e} \}$. Assuming that these sources act independently, then $T_3 \sim \exp(\lambda + 3\beta_e)$. As the number of infected neighbors of $V_3$ increases, its infection rate also increases. The extended contact process models a \emph{distributed} contagion scenario where all the infection sources compete to be the first to infect a healthy agent.

\item[Scaled SIS Process] \hfill \\
The scaled SIS process is parameterized by the exogenous infection rate, $\lambda$, healing rate, $\mu$, and the endogenous infection rate $\beta_s$. Consider the scenario in Figure~\ref{fig:contact}. Let $T_3$ be the random amount of time it takes for agent $V_3$ to become infected. As agent $V_3$ has three infected neighbors (i.e., $V_1, V_2, V_4$), the scaled SIS process assumes that $T_3 = \frac{1}{(\beta_s)^3}T \sim \exp(\lambda(\beta_s)^{3})$, where $T \sim \exp(\lambda)$ is the random amount of time a healthy agent becomes infected when it has no infected neighbors.  When $\beta_s > 1$,  as the number of infected neighbors of $V_3$ increases, its infection rate also increases. Unlike the extended contact process, the scaled SIS process assumes an \emph{aggregate} contagion scenario.

\end{description}

\section{Time-Asymptotic Behavior of the Extended Contact Process}\label{sec:timeasy}
For finite-size networks, unlike the contact process, the equilibrium distribution of the extended contact process is nontrivial. In this section, we show that, for a subclass of extended contact processes over \emph{arbitrary} network topology, this equilibrium distribution is well approximated by the equilibrium distribution of the scaled SIS process; for these processes, the time-asymptotic behavior of both processes are similar.

\begin{lemma}\label{lemma:approx}[Proof in Appendix~\ref{proof:lemma:approx}]
For any nonnegative integer $m$ from $0$ to $d_{\max}$, if 
\[
\Delta^2 << \frac{2}{d_{\max}(d_{\max} - 1)},
\]
then
\[
\frac{\lambda}{\mu}(1+\Delta)^{m} \approx  \frac{\lambda}{\mu} + \frac{\lambda}{\mu}\Delta m.
\]
\end{lemma}

Using Lemma~\ref{lemma:approx}, we can prove the following theorem.
\begin{theorem}\label{theorem:scaledcontactapprox}[Proof in Appendix~\ref{proof:theorem:scaledcontactapprox}]
Consider the extended contact process exogenous infection rate $\lambda$, healing rate $\mu$, and endogenous infection rate $\beta_e$, over a static, simple, connected, undirected network of arbitrary topology, $G$, with maximum degree $d_{\max}$. Let $\beta_e = \frac{\lambda}{\mu}\Delta$. If 
\[
\Delta^2 << \frac{2}{d_{\max}(d_{\max} - 1)},
\]
then the equilibrium distribution of the extended contact process is well approximated by 
\begin{equation}\label{eq:eqapprox}
\pi_{\scriptsize\textrm{approx}}(\mathbf{x})= \frac{1}{Z}\left( \frac{\lambda}{\mu}\right)^{1^T{\mathbf x}}  (1 + \Delta)^{\frac{{\mathbf x}^TA{\mathbf x}}{2}  }, \quad  \mathbf{x} \in \mathcal{X} ,
\end{equation}
where $A$ is the adjacency matrix of the network $G$, and $Z$ is the partition function. The approximate distribution, $\pi_{\scriptsize\textrm{approx}}(\mathbf{x})$, is the equilibrium distribution \eqref{eq:eqapprox} of a scaled SIS process over the same network $G$ with exogenous infection rate $\lambda$, healing rate $\mu$, and endogenous infection rate $\beta_s = 1+ \Delta$.
\end{theorem}

Theorem~\ref{theorem:scaledcontactapprox} gives an upperbound on the factor, $\Delta$, between the endogenous infection rate, $\beta_e$, and the ratio, $\frac{\lambda}{\mu}$, of the exogenous infection rate, $\lambda$, and the healing rate, $\mu$. This bound depends on the maximum degree of the underlying network topology. When $\beta_e$ is much smaller than $\frac{\lambda}{\mu}$, then the equilibrium distribution, $\pi_e(\mathbf{x})$, of the extended contact process is well approximated by that of an equivalent scaled SIS process. What does this imply about the extended contact process?

Recall that, for the extended contact process, all infection sources are \emph{independent}. Suppose that susceptible agent $i$ has one infected neighbor. Let $T_i^1 \sim \exp(\beta_e)$ be the random amount of time it takes for susceptible agent $i$ to be infected by this infected neighbor, and $T_i^e \sim \exp(\lambda)$ be the random amount of time it takes for susceptible agent $i$ to become infected by an exogenous source. The probability that the agent $i$ is infected by the exogenous source rather than by its infected neighbors is 
\begin{equation}\label{eq:singleinfection}
P(T_i^e \le T_i^1) = \frac{\lambda}{\beta_e + \lambda} = \frac{1}{\frac{\Delta}{\mu} + 1},
\end{equation}
since $\beta_e = \frac{\lambda}{\mu}\Delta$. (See Appendix~\ref{expmin} for a review regarding functions of exponentially distributed random variables.) Suppose that susceptible agent $i$ has multiple (i.e., $m > 1$) infected neighbors. The probability that agent $i$ is infected by the exogenous source rather than by its infected neighbors is
\begin{equation}\label{eq:multinfection}
P(T_i^e \le \min\{T_i^1, \ldots T_i^m\}) = \frac{\lambda}{m\beta_e + \lambda} = \frac{1}{\frac{m\Delta}{\mu} + 1}.
\end{equation}

Without loss of generality, let $\mu = 1$. According to Theorem~\ref{theorem:scaledcontactapprox}, the scaled SIS process is a valid approximation for the extended contact process when $\Delta$ is small. In this case, according to \eqref{eq:singleinfection} and \eqref{eq:multinfection}, the probability that the source of infection is exogenous rather than endogenous is high; infection due to contagion from infected neighbor is rare but not impossible.


\section{Experimental Simulations}\label{sec:numerical}

We showed when the extended contact process can be well approximated by the scaled SIS process. We confirm Theorem~\ref{theorem:scaledcontactapprox} with numerical simulations. Further, we show that this upperbound is conservative. Below it, the equilibrium distribution of the extended contact process. $\pi_e(\mathbf{x})$, for arbitrary network topology is well approximated by the equilibrium distribution, $\pi_{\scriptsize\textrm{approx}}(\mathbf{x})$, of a scaled SIS process. However, depending on the underlying network topology, the approximation may still remain accurate ($< 0.1$ deviation) for extended contact processes with parameters \emph{away from} the bound.

\subsection{Setup}
We will compare the true equilibrium distribution, $\pi_e(\mathbf{x})$, of the extended contact process, with infection and healing rates $\left(\lambda, \mu, \beta_e= \frac{\lambda}{\mu}\Delta\right)$ over network $G$, with the approximation distribution, $\pi_{\scriptsize\textrm{approx}}(\mathbf{x})$. The true distribution, $\pi_e(\mathbf{x})$, is found numerically by forming the transition rate matrix, $\mathbf{Q}_e$, according to \eqref{eq:contacthealrate} and \eqref{eq:extendedcontactinfectrate} and solving for the left eigenvector of $\mathbf{Q}_e$ corresponding to eigenvalue $0$. The approximate distribution, $\pi_{\scriptsize\textrm{approx}}(\mathbf{x})$, is obtained from the closed-form equation according to Theorem~\ref{theorem:scaledcontactapprox}
\[
\pi_{\scriptsize\textrm{approx}}(\mathbf{x})= \frac{1}{Z}\left( \frac{\lambda}{\mu}\right)^{1^T{\mathbf x}}  (1 + \Delta)^{\frac{{\mathbf x}^TA{\mathbf x}}{2}  }, \quad  \mathbf{x} \in \mathcal{X}.
\]

We solve for $\pi_e(\mathbf{x})$ and $\pi_{\scriptsize\textrm{approx}}(\mathbf{x})$ for different $\frac{\lambda}{\mu}$ values and different $\Delta$ values, both below and above the upperbound, 
\[
\Delta_u =\sqrt{\frac{2}{d_{\max}(d_{\max}-1)}}.
\]
To quantify the difference between the exact and the approximation equilibrium distribution, $\pi_e(\mathbf{x})$ and $\pi_{\scriptsize\textrm{approx}}(\mathbf{x})$, we use the total variation distance (TVD) \cite{levin2009markov}:
\begin{equation}\label{eq:TVD}
\text{TVD}(\pi_{e}, \pi_{\scriptsize\textrm{approx}}) = \frac{1}{2}\sum_{\mathbf{x} \in \mathcal{X}}|\pi_{e}(\mathbf{x}) - \pi_{\scriptsize\textrm{approx}}(\mathbf{x})|.
\end{equation}
When the two distributions are equal, TVD is 0. The maximum TVD between any two probability distributions over the same support is $1$.

As the true distribution of the extended contact process, $\pi_e(\mathbf{x})$, is obtained by solving the zero eigenvalue-eigenvector problem of $\mathbf{Q}_e$, which is a $2^N \times 2^N$ matrix, we are restricted to simulating examples with small networks of size $N$. We consider six $16$-node networks (see Figure~\ref{fig:networkexample}) with different maximum degree, $d_{\max}$, corresponding to different upperbounds $\Delta_u$. Networks A and B have the smallest possible maximum degree of any connected graph ($d_{\max}=2$); they have the largest possible upperbound ($\Delta_u =1$). Network F has the largest maximum degree of the networks studied ($d_{\max}=15$) and has the smallest upperbound ($\Delta_u = 0.098$). 

In Matlab, on a Microsoft Azure cloud virtual machine with 2.6GHz Intel Xeon E5-2670 and 56GB of RAM, for a 16-node network, it takes approximately $2$ secs to generate the sparse transition rate matrix $\mathbf{Q}_e$ and $460$ secs to solve for the eigenvector corresponding to the $0$ eigenvalue. For a 20-node network, it takes approximately $30$ secs to generate the transition rate matrix $\mathbf{Q}_e$; we receive an OUT-OF-MEMORY error when computing the eigenvector.

\begin{figure}[ht]      
      \centering
          \begin{subfigure}[b]{0.25\textwidth}
                \centering
                \includegraphics[width=\textwidth]{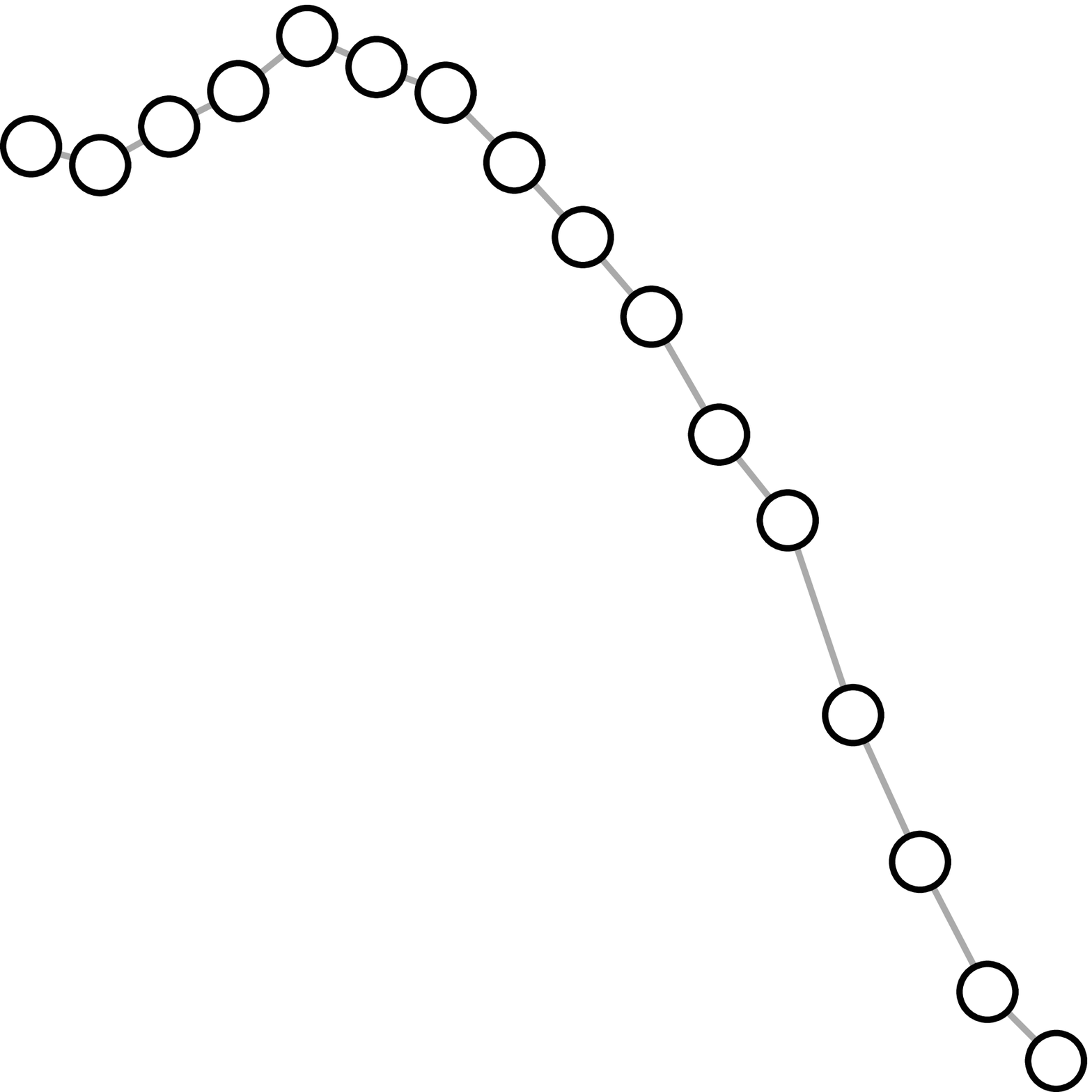}
                \caption{Network A\\($d_{\max} = 2, \Delta_u = 1$)}
                \label{fig:16Path}
        \end{subfigure}
	  \enspace
        \begin{subfigure}[b]{0.25\textwidth}
                \centering
                \includegraphics[width=\textwidth]{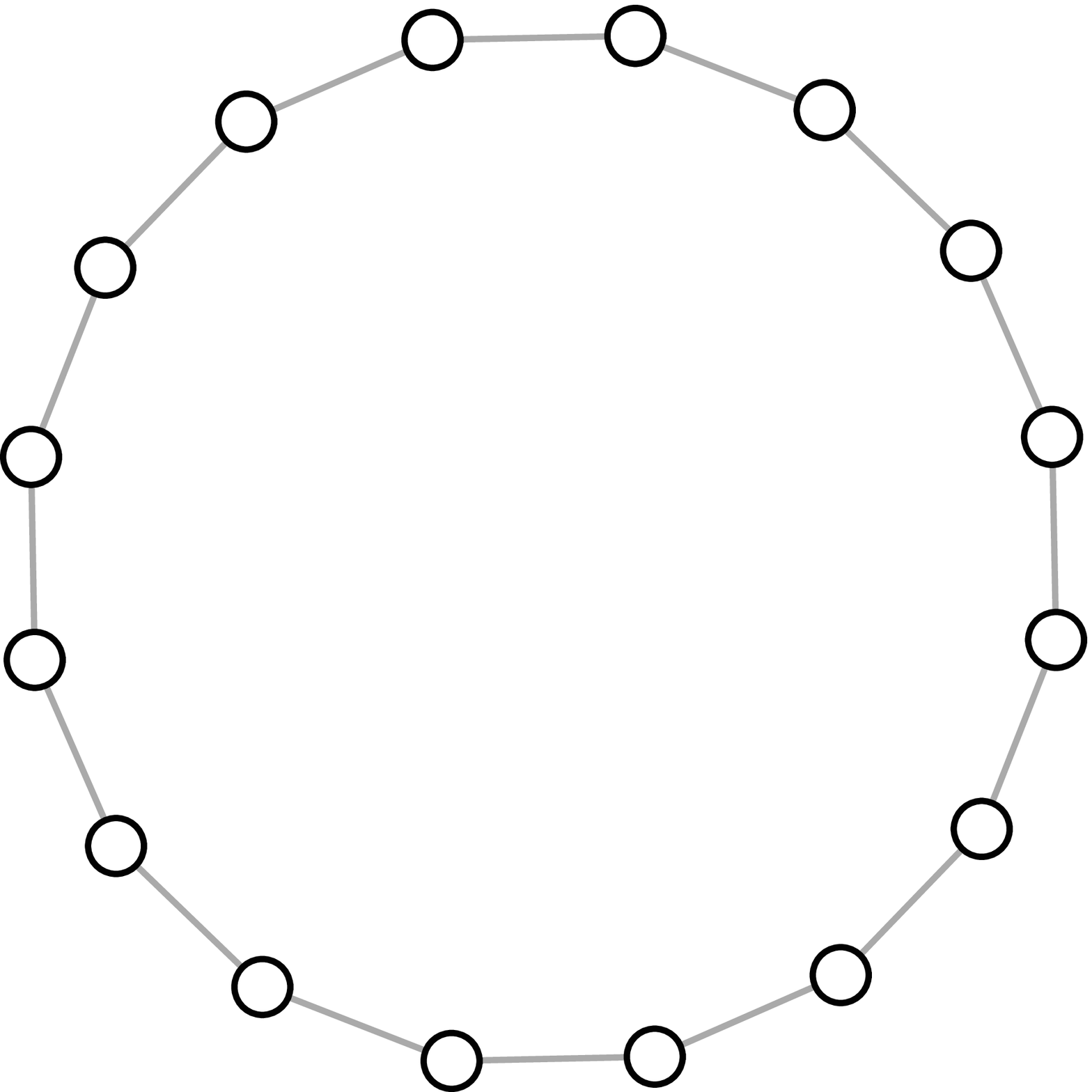}
                \caption{Network B\\($d_{\max} = 2, \Delta_u = 1$)}
                \label{fig:16Cycle}
        \end{subfigure}%
         \enspace
          \begin{subfigure}[b]{0.25\textwidth}
                \centering
                \includegraphics[width=\textwidth]{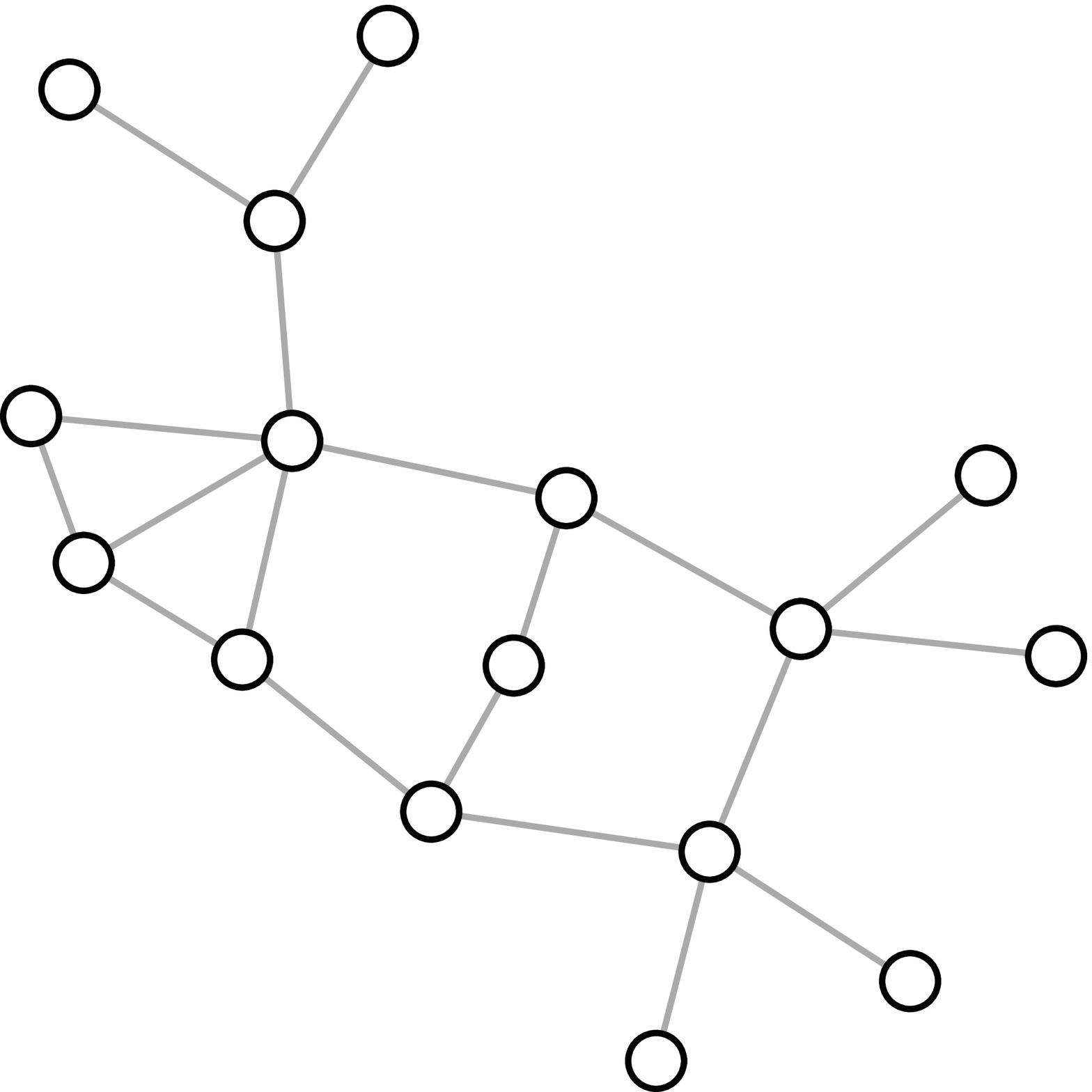}
                \caption{Network C\\($d_{\max} = 5, \Delta_u = 0.32$)}
                \label{fig:16ERa}
        \end{subfigure}
        
        \hfill \\
          \begin{subfigure}[b]{0.25\textwidth}
                \centering
                \includegraphics[width=\textwidth]{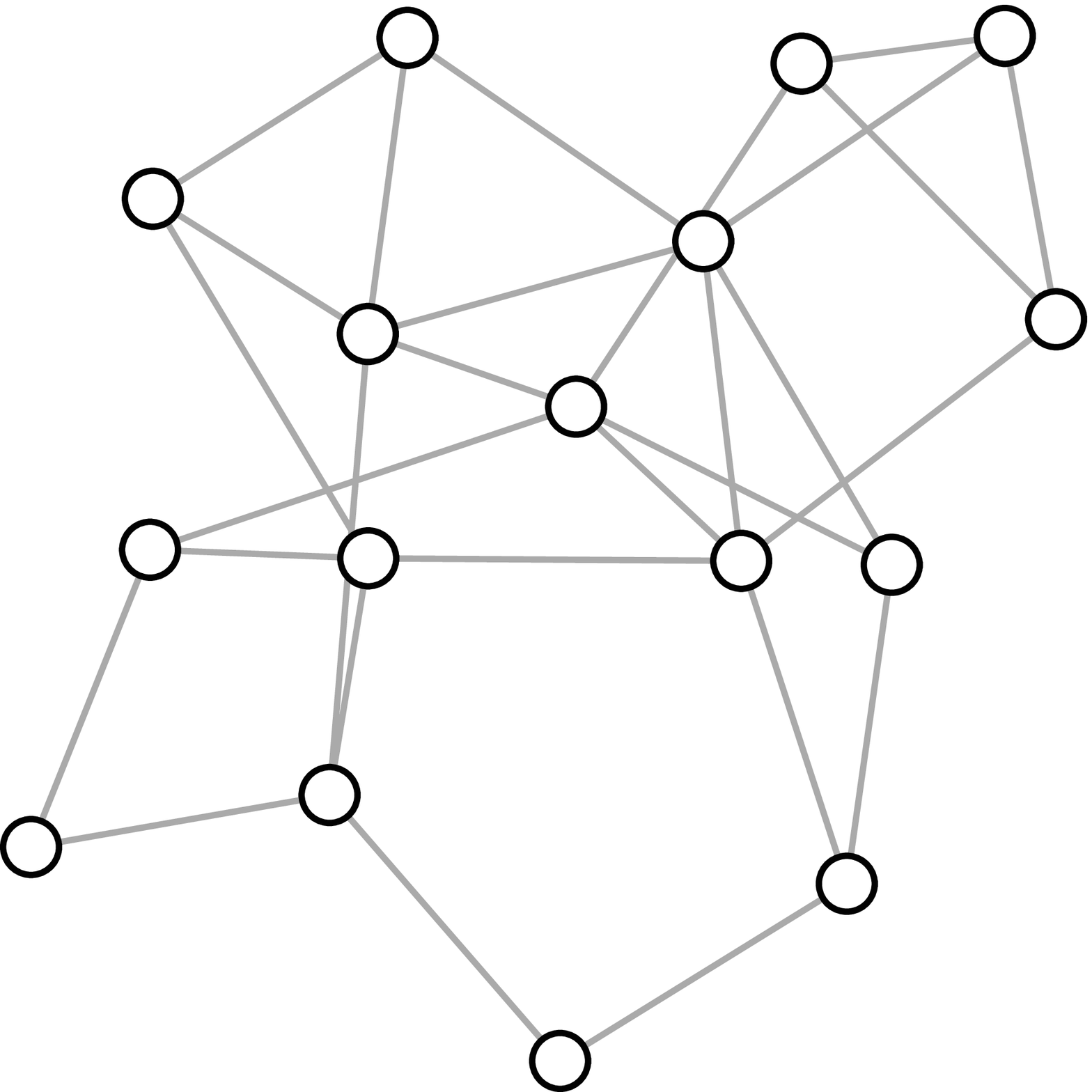}
                \caption{Network D\\($d_{\max} = 5, \Delta_u = 0.32$)}
                \label{fig:16ERb}
        \end{subfigure}
	 \enspace 
         \begin{subfigure}[b]{0.25\textwidth}
                \centering
                \includegraphics[width=\textwidth]{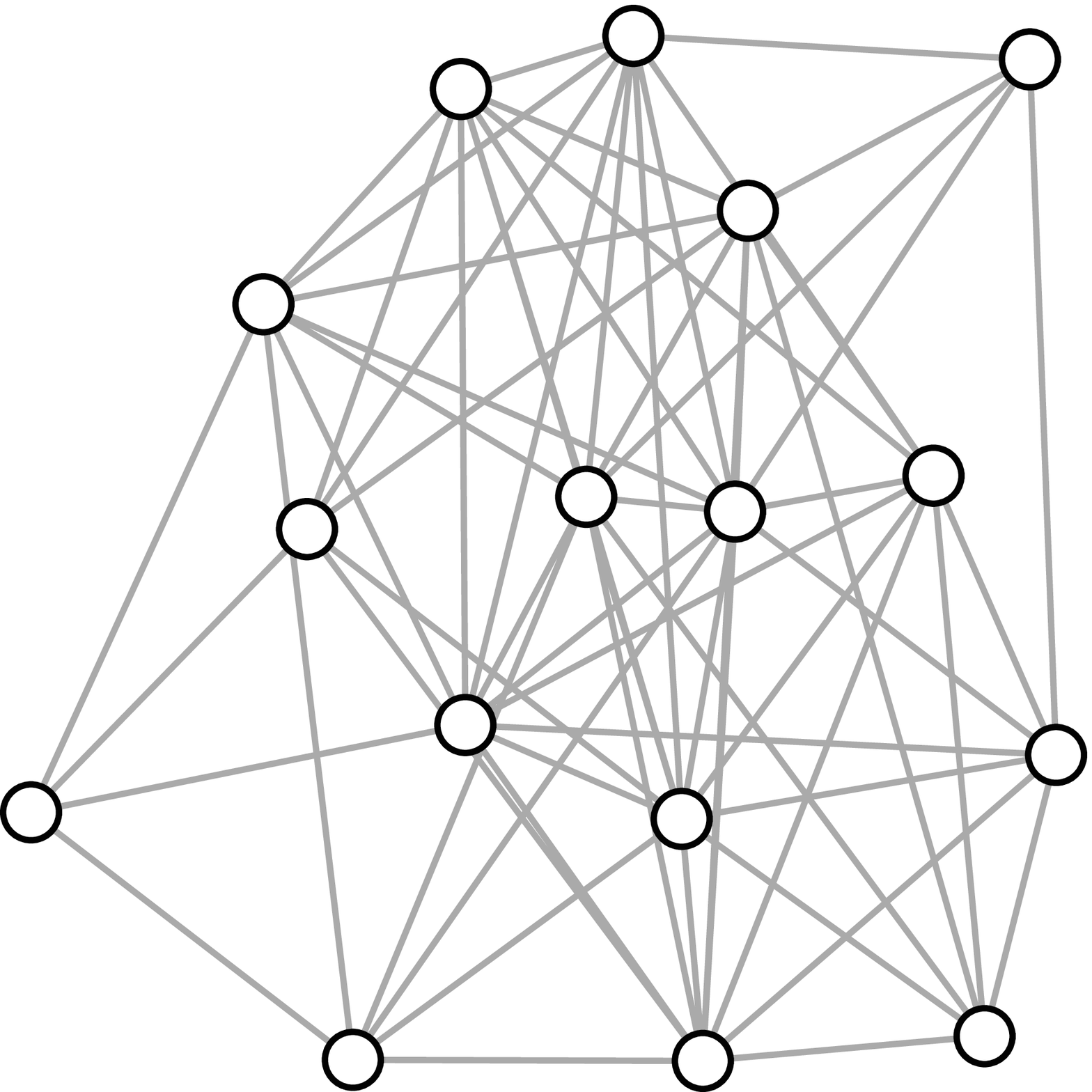}
                \caption{Network E\\($d_{\max} = 11, \Delta_u = 0.135$)}
                \label{fig:16WSa}
        \end{subfigure}
       \enspace 
         \begin{subfigure}[b]{0.25\textwidth}
                \centering
                \includegraphics[width=\textwidth]{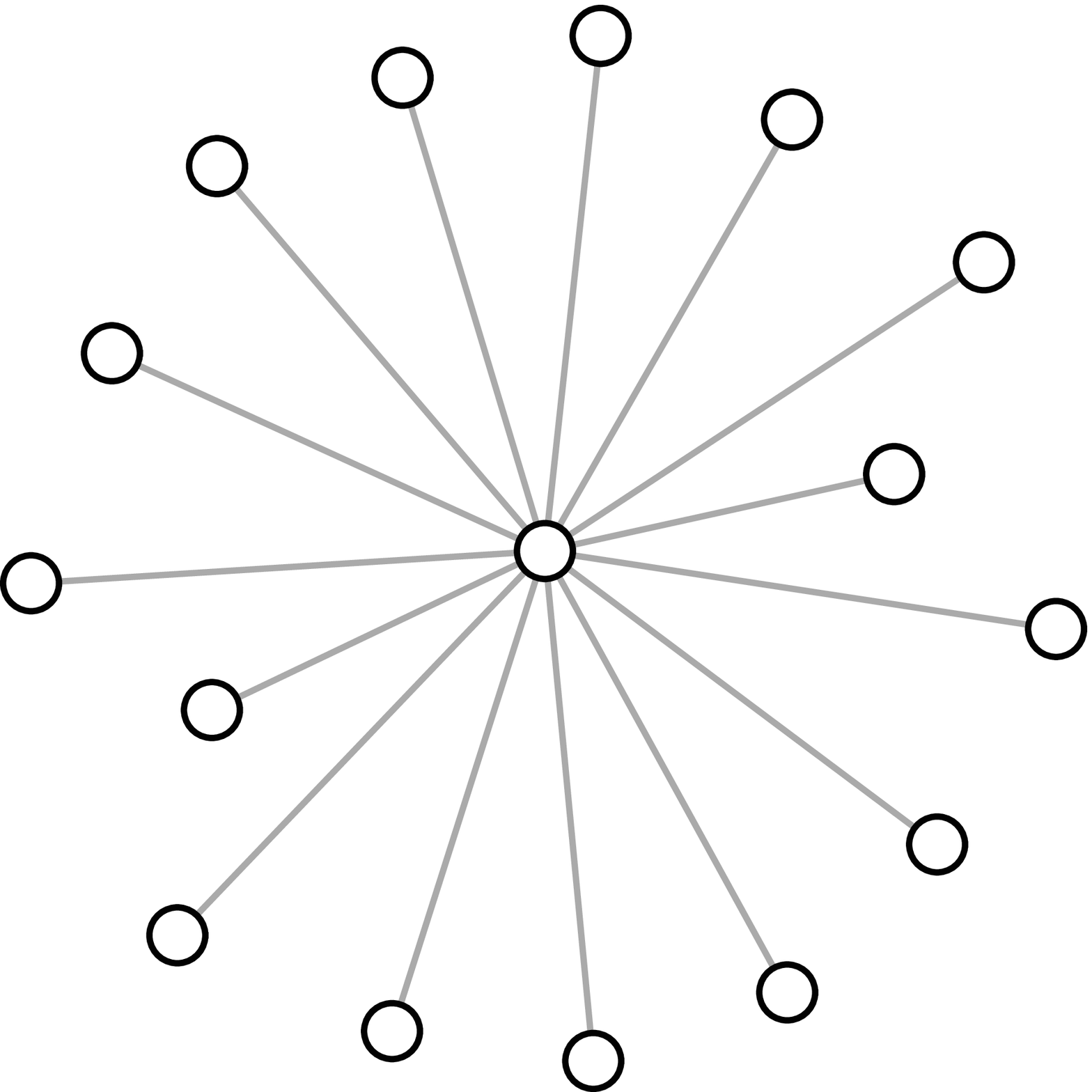}
                \caption{Network F\\($d_{\max} = 15, \Delta_u = 0.098$)}
                \label{fig:16Star}
        \end{subfigure}

	\caption{Different Network Topologies with Different Maximum Degree}\label{fig:networkexample}
 \end{figure}

\subsection{Results: $\pi_e(\mathbf{x})$ and $\pi_{\scriptsize\textrm{approx}}(\mathbf{x})$}
To provide intuition on the quality of the approximations for different TVDs, we plot in Figure~\ref{fig:eqvsplotsinside} the true equilibrium distribution, $\pi_e(\mathbf{x})$, of the extended contact process together with the approximate equilibrium distribution, $\pi_{\scriptsize\textrm{approx}}(\mathbf{x})$. The Y-axis displays both equilibrium distributions; we use log scaling for better visualization. The $2^{16}$ network configurations are on the X-axis. The configurations are ordered such that high probability configurations in $\pi_e(\mathbf{x})$ are in the center. 

Figure~\ref{fig:eqvsplotsinside} shows $\pi_e(\mathbf{x})$ and $\pi_{\scriptsize\textrm{approx}}(\mathbf{x})$ and their corresponding TVD, see \eqref{eq:TVD}, for the six different network topologies (see Figure~\ref{fig:networkexample}) for parameters $\frac{\lambda}{\mu} = 0.7, \Delta = 0.0023$. This value of $\Delta$ is much smaller than the upperbound, $\Delta_u$, for all the networks. The equilibrium distribution, $\pi_e(\mathbf{x})$, of the extended contact process is well approximated (i.e., the TVD is on the order of $10^{-5}$ or smaller) by equilibrium distribution, $\pi_{\scriptsize\textrm{approx}}(\mathbf{x})$, of the scaled SIS process. Note that this value of TVD is over $2^{16}$ configurations; so the actual divergence for any configuration is very small. The two distributions are almost identical for all the networks. 
%

%

\begin{figure}[htbp]      
      \centering
             \begin{subfigure}[b]{0.45\textwidth}
                \centering
                \includegraphics[width=\textwidth]{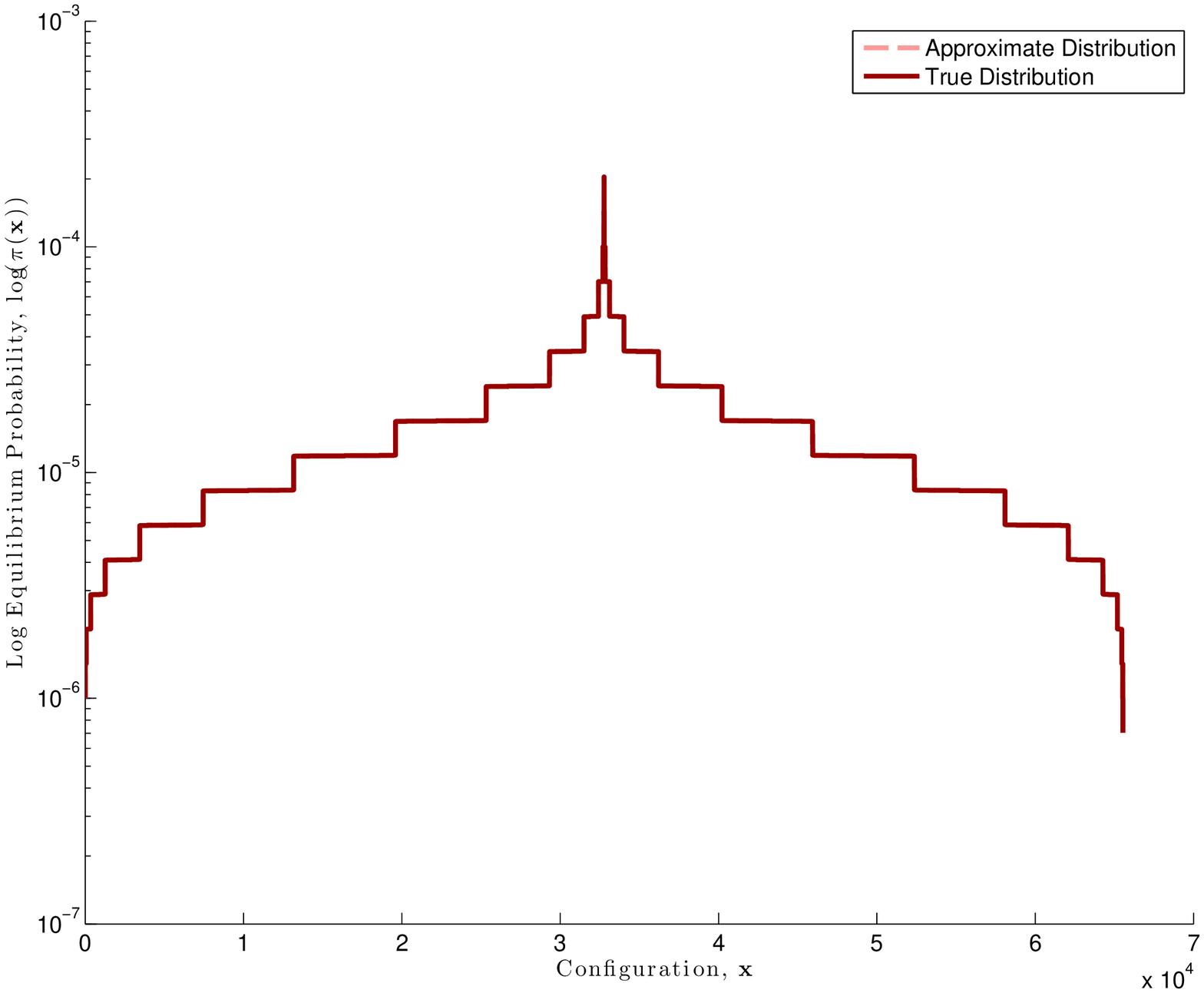}
                \caption{Network A (TVD = $1.0384\times 10^{-6}$)}
                \label{fig:16PathTVDeq2}
        \end{subfigure}%
        ~ 
        \begin{subfigure}[b]{0.45\textwidth}
                \centering
                \includegraphics[width=\textwidth]{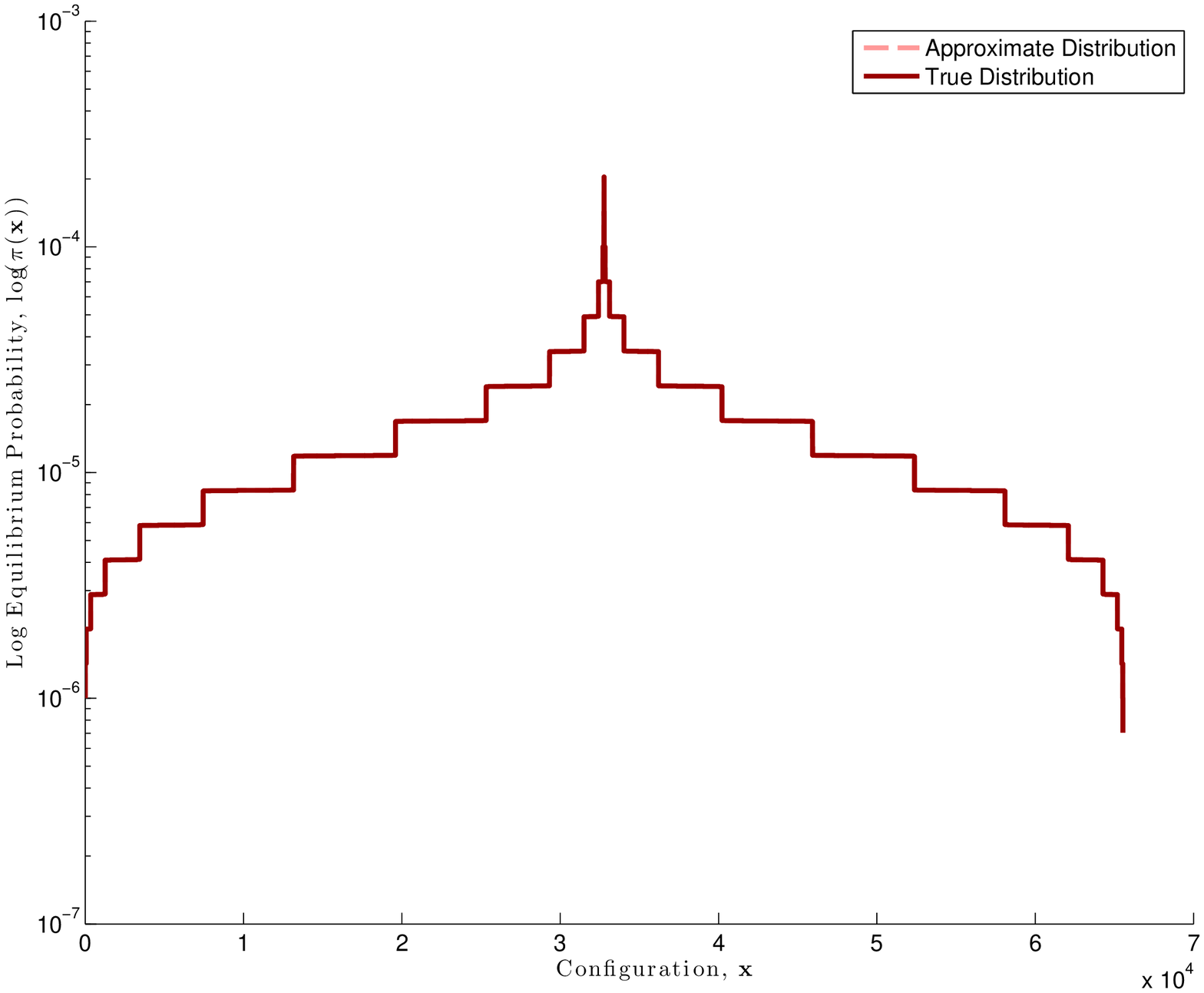}
                \caption{Network B (TVD = $1.1236 \times 10^{-6}$)}
                \label{fig:16CycleTVDeq2}
        \end{subfigure}%

        \begin{subfigure}[b]{0.45\textwidth}
                \centering
                \includegraphics[width=\textwidth]{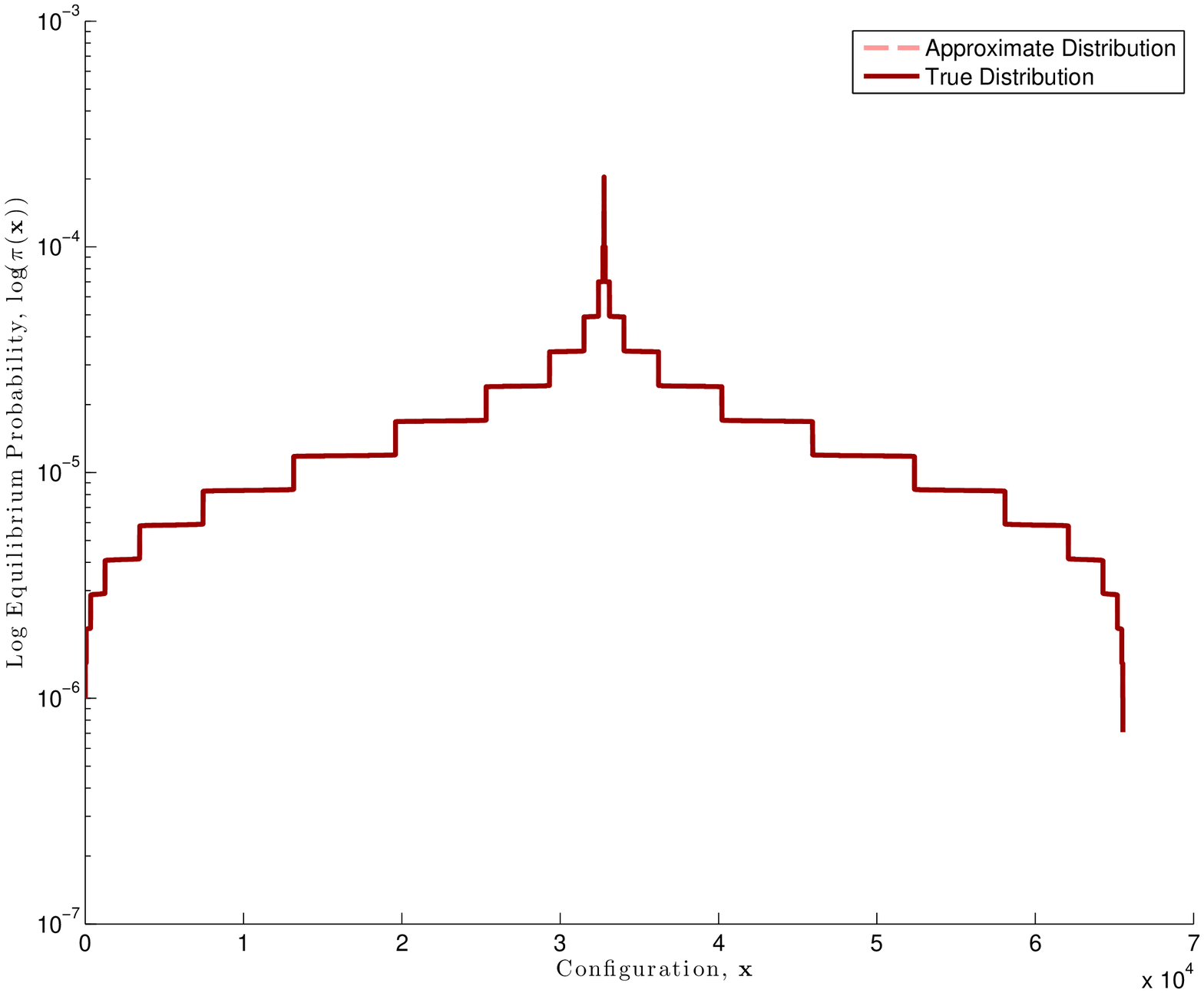}
                \caption{Network C (TVD = $3.2392\times 10^{-6}$)}
                \label{fig:16ERaTVDeq2}
        \end{subfigure}%
            ~ 
         \begin{subfigure}[b]{0.45\textwidth}
                \centering
                \includegraphics[width=\textwidth]{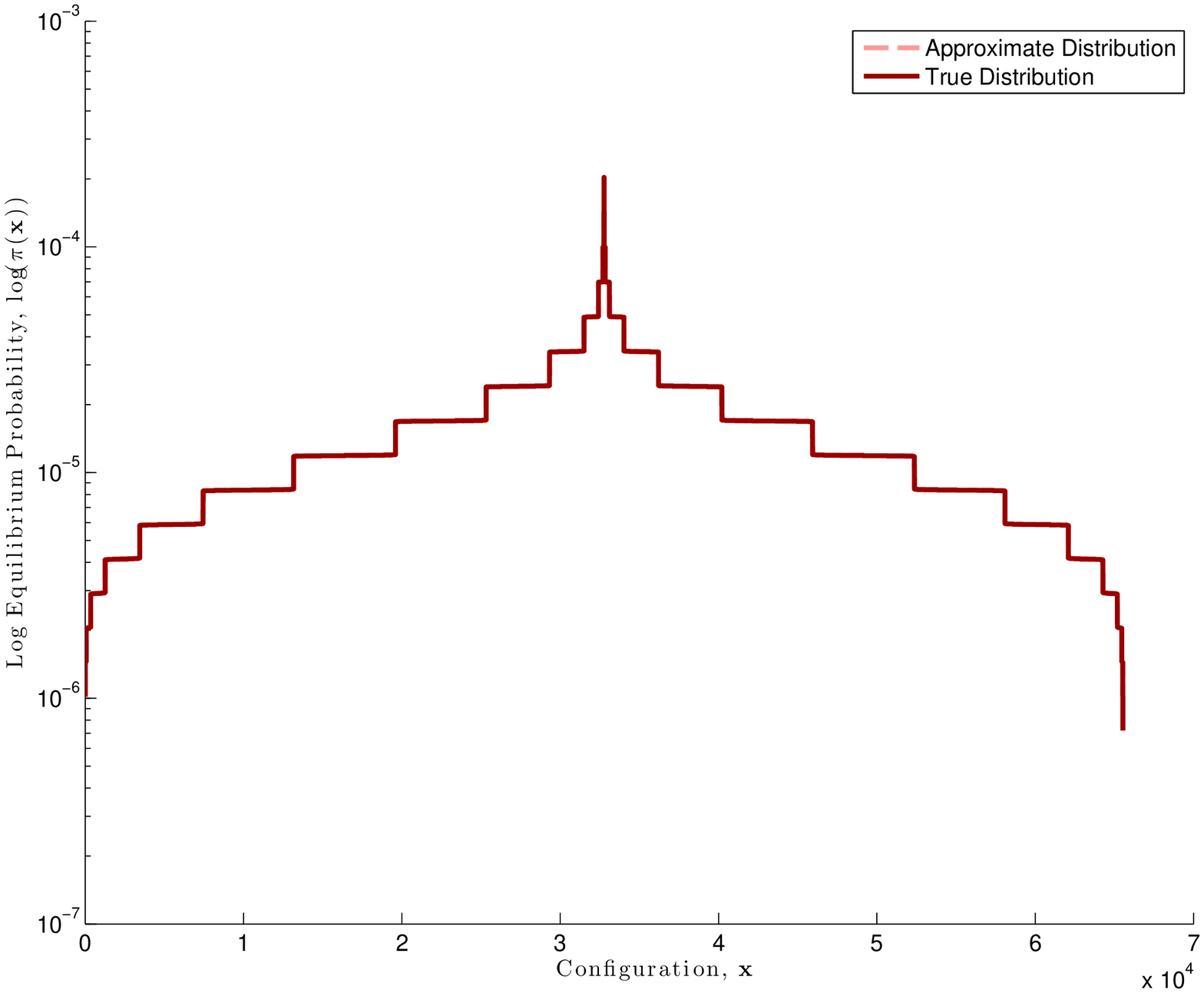}
                \caption{Network D (TVD = $5.0208\times 10^{-6}$)}
                \label{fig:16ERbTVDeq2}
        \end{subfigure}%
       
        \begin{subfigure}[b]{0.45\textwidth}
                \centering
                \includegraphics[width=\textwidth]{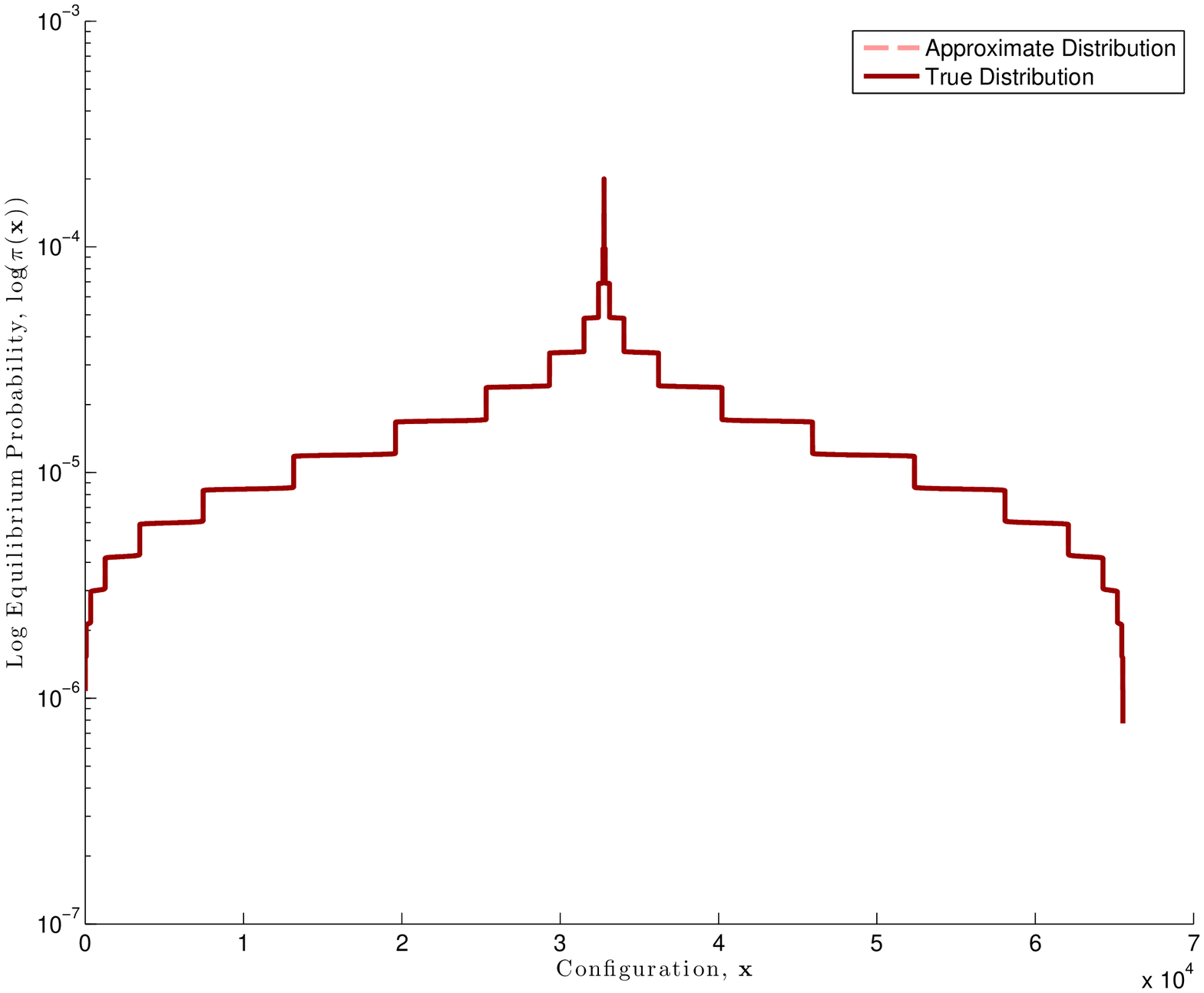}
                \caption{Network E (TVD = $2.7487\times 10^{-5}$)}
                \label{fig:16WSaTVDeq2}
        \end{subfigure}%
         ~ 
        \begin{subfigure}[b]{0.45\textwidth}
                \centering
                \includegraphics[width=\textwidth]{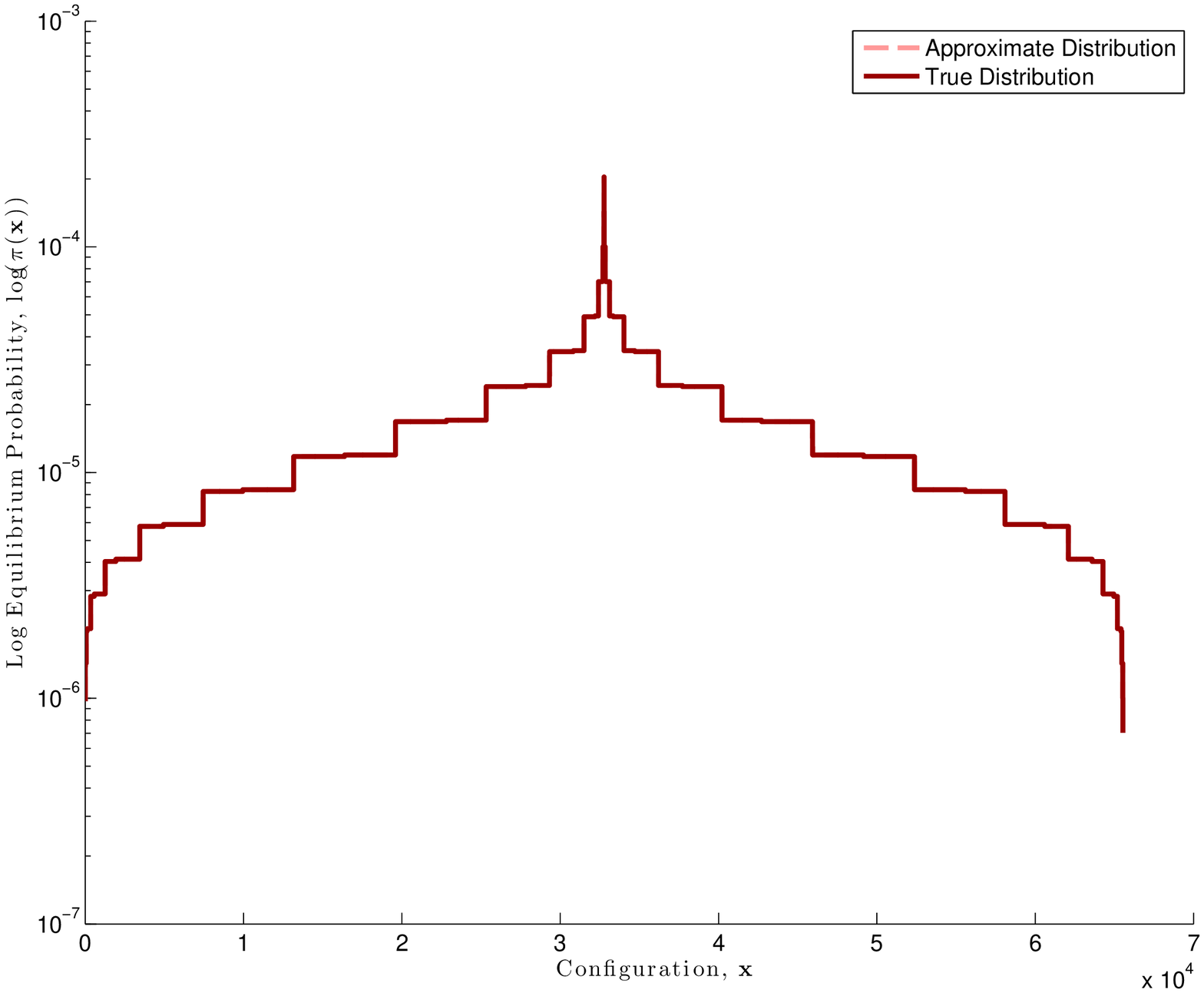}
                \caption{Network F (TVD = $2.2729\times 10^{-5}$)}
                \label{fig:16StarTVDeq2}
        \end{subfigure}%

	\caption{$\pi_e(\mathbf{x})$ and $\pi_{\scriptsize\textrm{approx}}(\mathbf{x})$ when $\frac{\lambda}{\mu} = 0.7, \Delta = 0.0023$}\label{fig:eqvsplotsinside}
 \end{figure}	
 

We also considered the case when the condition of Theorem~\ref{theorem:scaledcontactapprox} is not satisfied. Figure~\ref{fig:eqvsplotsoutside} shows $\pi_e(\mathbf{x})$ and $\pi_{\scriptsize\textrm{approx}}(\mathbf{x})$ and their corresponding TVD for parameters $\frac{\lambda}{\mu} = 0.7, \Delta = 1.0496$. In this case, the value of $\Delta$ is above the upperbound $\Delta_u$ for all the networks in Figure~\ref{fig:networkexample}. As we expect, TVD is larger when compared to the TVD for processes with $\Delta$ well below $\Delta_u$ (see Figure~\ref{fig:eqvsplotsinside}). Again, for the same infection and healing rates, different networks have different TVD values.

For Networks A and B, the deviation between the true and approximate equilibrium distribution, $0.1073$ and $0.1186$,  respectively, is relatively small. We see from Figure~\ref{fig:16PathTVDeq} and Figure~\ref{fig:16CycleTVDeq} that many configurations have similar equilibrium probability for both distributions. For Network D and E, TVDs are $0.4798$ and $0.7898$, respectively. These figures show that the approximate distribution tends to overestimate the probability of highly probable configurations but underestimates the low probability configurations. However, there is good correlation between the relative ordering of configurations in both distributions; configurations that are highly probable in $\pi_e(\mathbf{x})$ are also highly probable in $\pi_{\scriptsize\textrm{approx}}(\mathbf{x})$.


\begin{figure}[htbp]      
      \centering
             \begin{subfigure}[b]{0.45\textwidth}
                \centering
                \includegraphics[width=\textwidth]{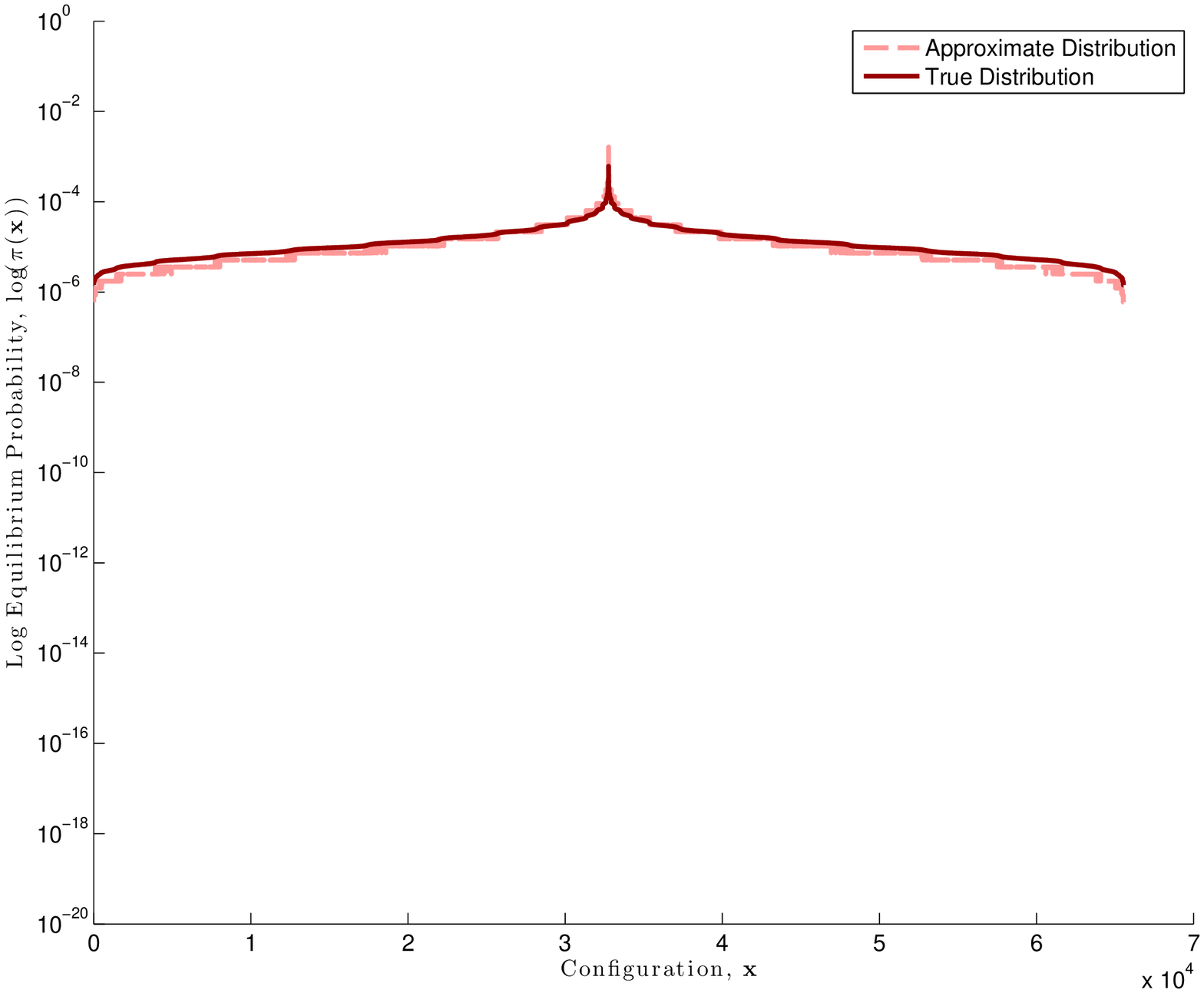}
                \caption{Network A (TVD = $0.1073$)}
                \label{fig:16PathTVDeq}
        \end{subfigure}%
        ~ 
        \begin{subfigure}[b]{0.45\textwidth}
                \centering
                \includegraphics[width=\textwidth]{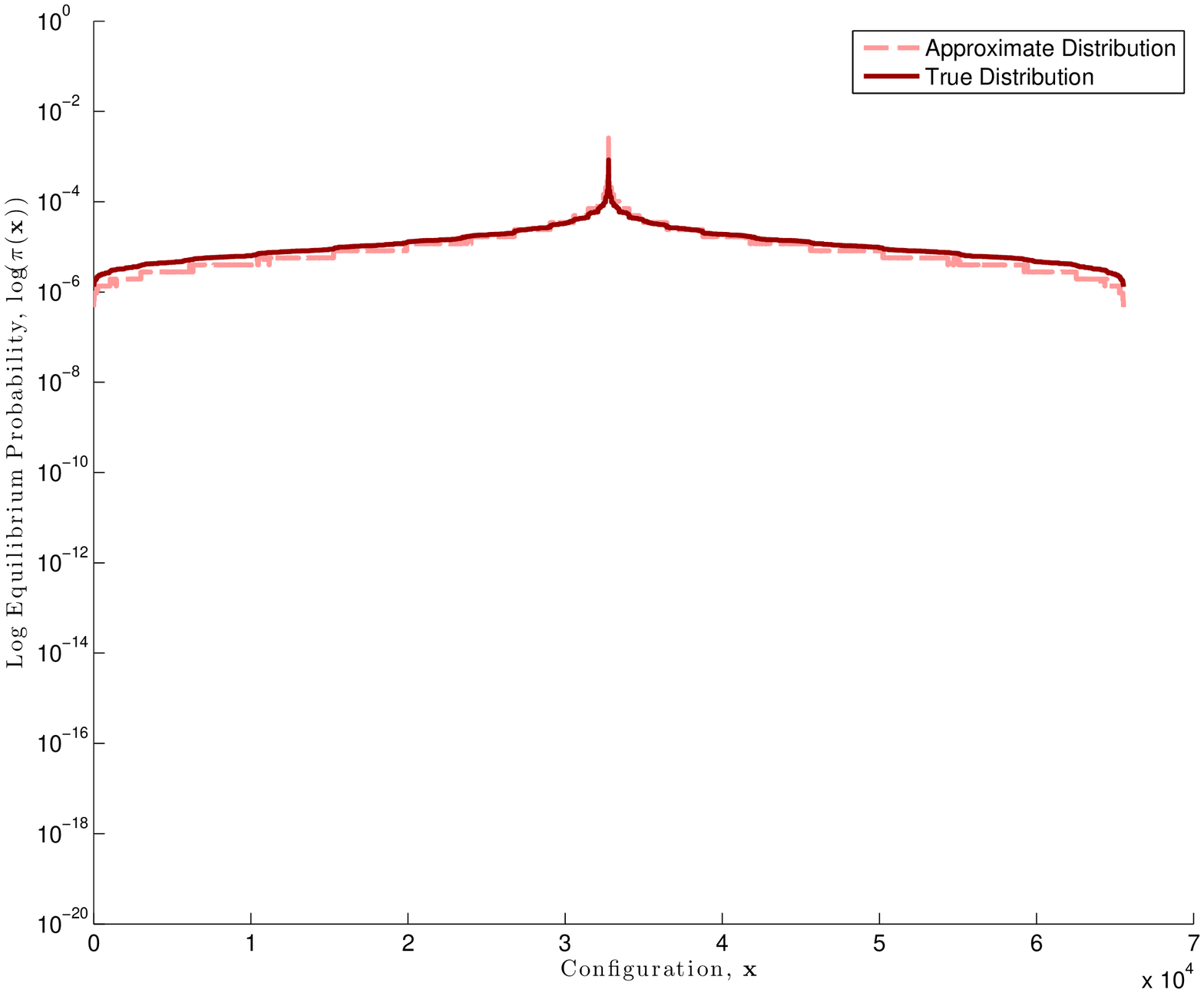}
                \caption{Network B (TVD = $0.1186$)}
                \label{fig:16CycleTVDeq}
        \end{subfigure}%

        \begin{subfigure}[b]{0.45\textwidth}
                \centering
                \includegraphics[width=\textwidth]{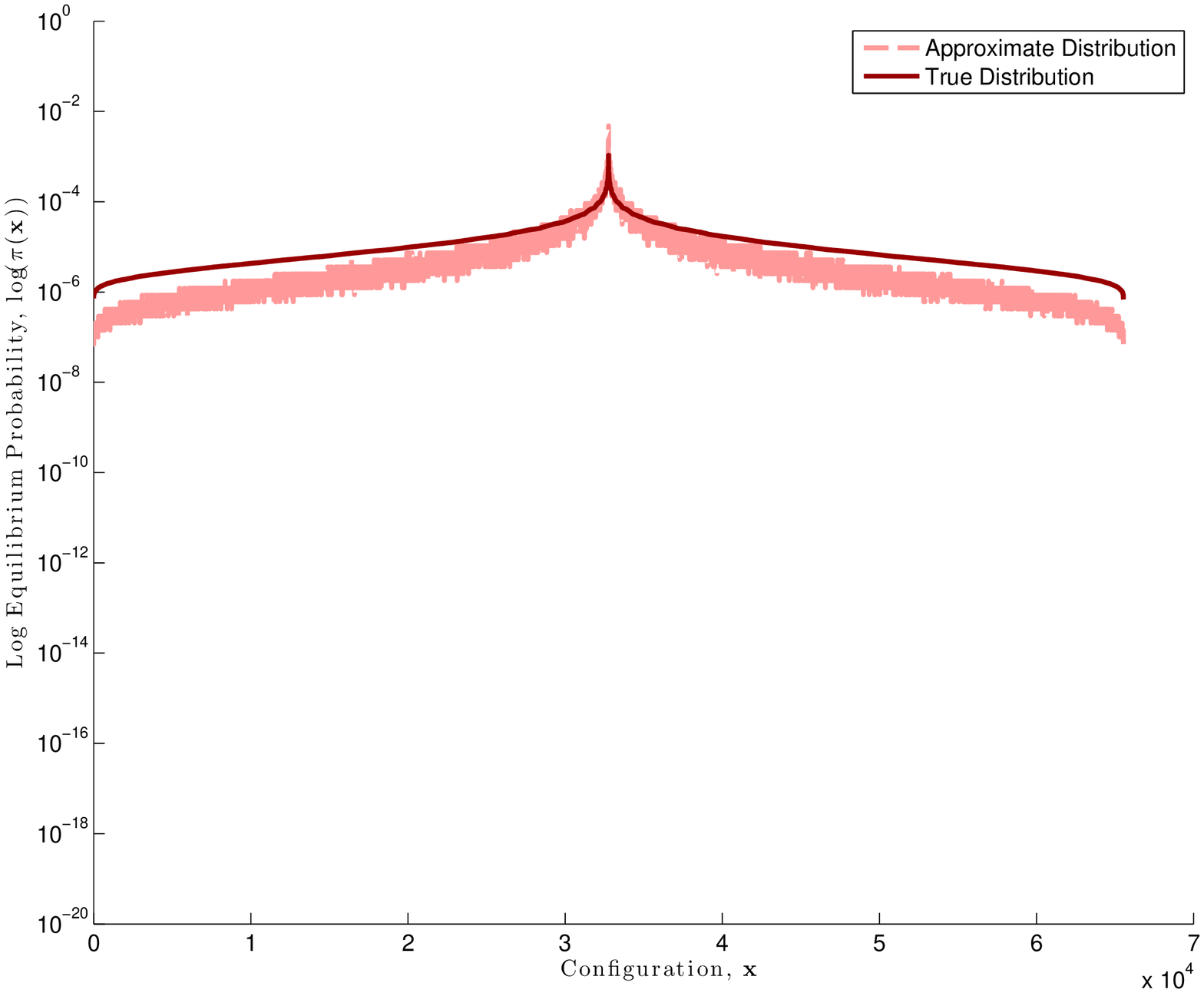}
                \caption{Network C (TVD = $0.294$)}
                \label{fig:16ERaTVDeq}
        \end{subfigure}%
            ~ 
         \begin{subfigure}[b]{0.45\textwidth}
                \centering
                \includegraphics[width=\textwidth]{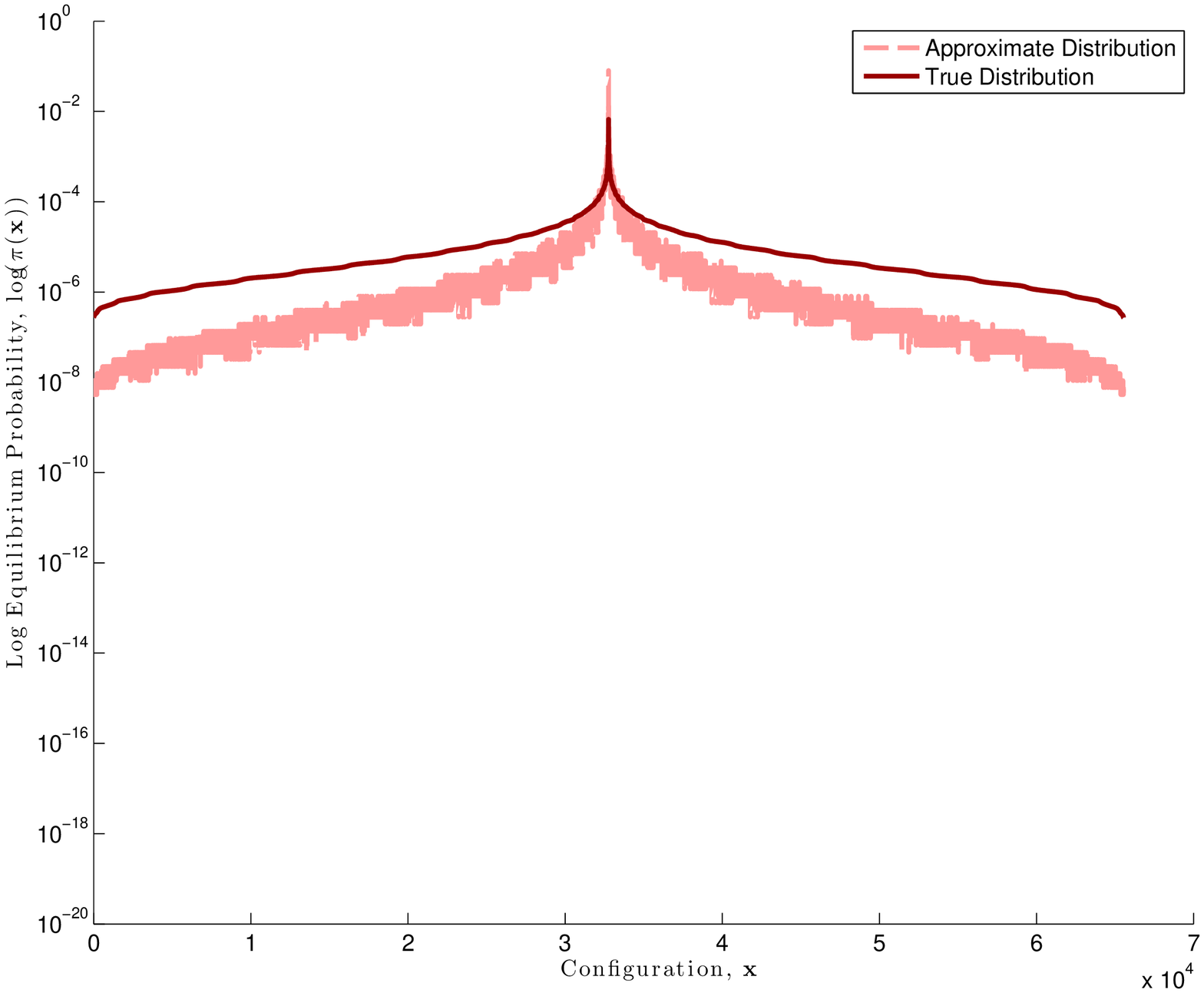}
                \caption{Network D (TVD = $0.4798$)}
                \label{fig:16ERbTVDeq}
        \end{subfigure}%
       
        \begin{subfigure}[b]{0.45\textwidth}
                \centering
                \includegraphics[width=\textwidth]{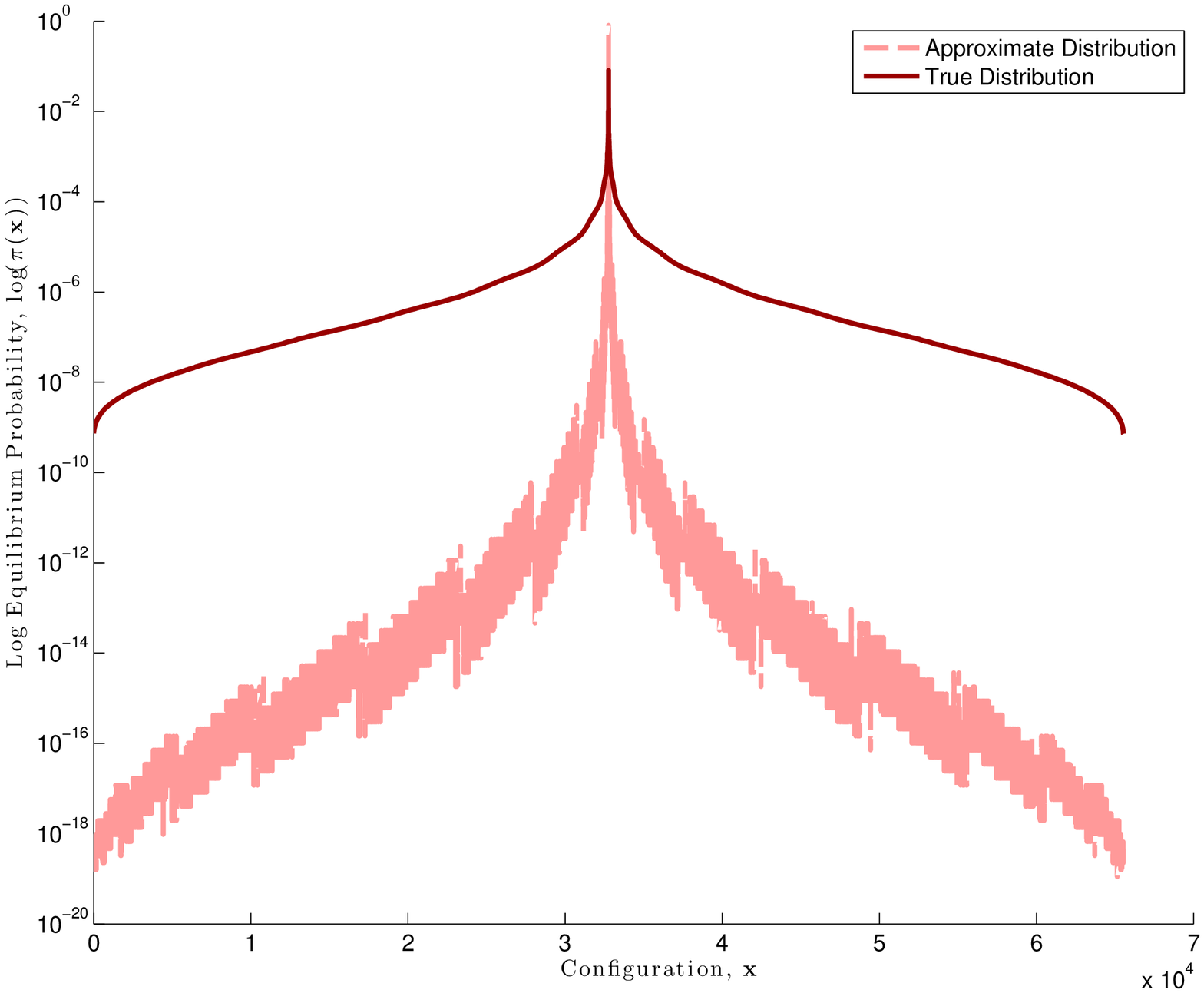}
                \caption{Network E (TVD = $0.7898$)}
                \label{fig:16WSaTVDeq}
        \end{subfigure}%
         ~ 
        \begin{subfigure}[b]{0.45\textwidth}
                \centering
                \includegraphics[width=\textwidth]{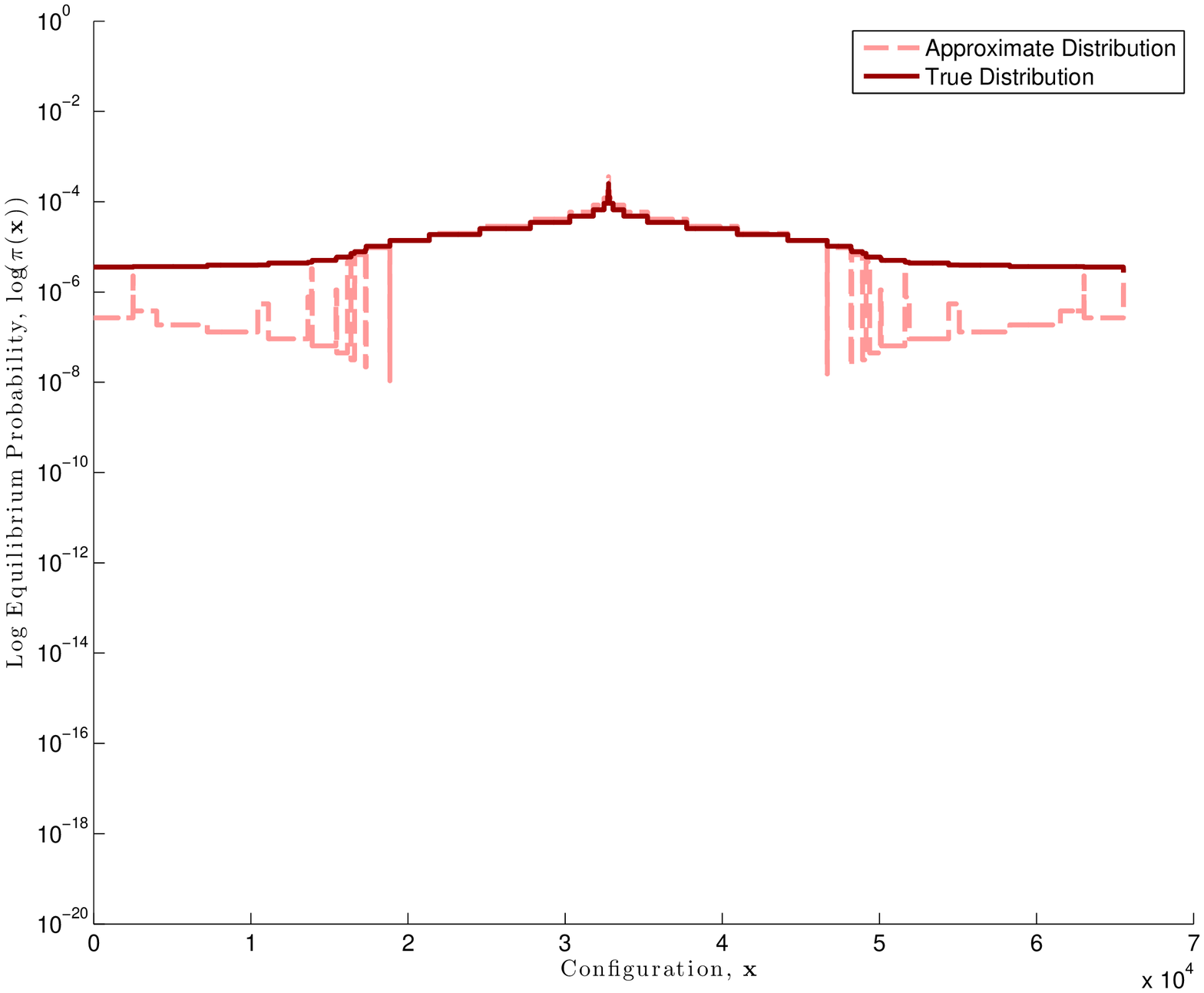}
                \caption{Network F (TVD = $0.1330$)}
                \label{fig:16StarTVDeq}
        \end{subfigure}%

	\caption{$\pi_e(\mathbf{x})$ and $\pi_{\scriptsize\textrm{approx}}(\mathbf{x})$ when $\frac{\lambda}{\mu} = 0.7, \Delta = 1.0496$}\label{fig:eqvsplotsoutside}
 \end{figure}	

\subsection{Results: TVD vs. $\Delta$ and $\frac{\lambda}{\mu}$}
We consider here the approximation of $\pi_e(\mathbf{x})$ by $\pi_{\scriptsize\textrm{approx}}(\mathbf{x})$ as the infection and healing rates change. Like the scaled SIS process, we can interpret the extended contact process as consisting of a \emph{topology-independent} process parameterized by $\frac{\lambda}{\mu}$\textemdash it is topology independent because the exogenous infection rate $\lambda$ and the healing rate $\mu$ are identical for all the agents in the network\textemdash and a \emph{topology-dependent} process parameterized by the endogenous infection rate $\beta$. When $\frac{\lambda}{\mu}$ is large, the topology-independent process exerts a larger effect on the equilibrium behavior of the network processes.

Figure~\ref{fig:networkexample1TVD} shows the TVD between the equilibrium distribution, $\pi_e(\mathbf{x})$, of the extended contact process and the approximate distribution, $\pi_{\scriptsize\textrm{approx}}(\mathbf{x})$, for different $\frac{\lambda}{\mu}$ and $\Delta$ values both below and well above the threshold $\Delta_u$. Figure~\ref{fig:networkexample1TVD} considers the six different network topologies in Figure~\ref{fig:networkexample}. We plot $\Delta$ along the X-axis and the TVD between $\pi_e(\mathbf{x})$ and $\pi_{\scriptsize\textrm{approx}}(\mathbf{x})$ along the Y-axis. Different curves in each figure correspond to equilibrium distributions with different  $\frac{\lambda}{\mu}$ values.


For the same $\Delta$ value, we observe that larger $\frac{\lambda}{\mu}$ correspond to smaller TVD. This holds for all the networks. Also, as we expect, for $\Delta << \Delta_u$, TVD is negligible for all the networks. The deviation between the true equilibrium distribution and the approximation increases as $\Delta$ moves toward $\Delta_u$; the rate of this increase differs for different topologies. Surprisingly, this increase is \emph{not} monotonic for \emph{all} network topologies. As $\Delta$ increases to values larger than $\Delta_u$, TVD may actually decrease. We observe this decrease in TVD for both Network E in Figure~\ref{fig:16WSaTVD} and Network F in Figure~\ref{fig:16StarTVD}. In particular, Network F, which has the largest maximum degree of all the six networks, has relatively small deviation between $\pi_{e}(\mathbf{x})$ and $\pi_{\scriptsize\textrm{approx}}(\mathbf{x})$ compare to the other network topologies.

%

\begin{figure}[htbp]      
      \centering
             \begin{subfigure}[b]{0.4\textwidth}
                \centering
                \includegraphics[width=\textwidth]{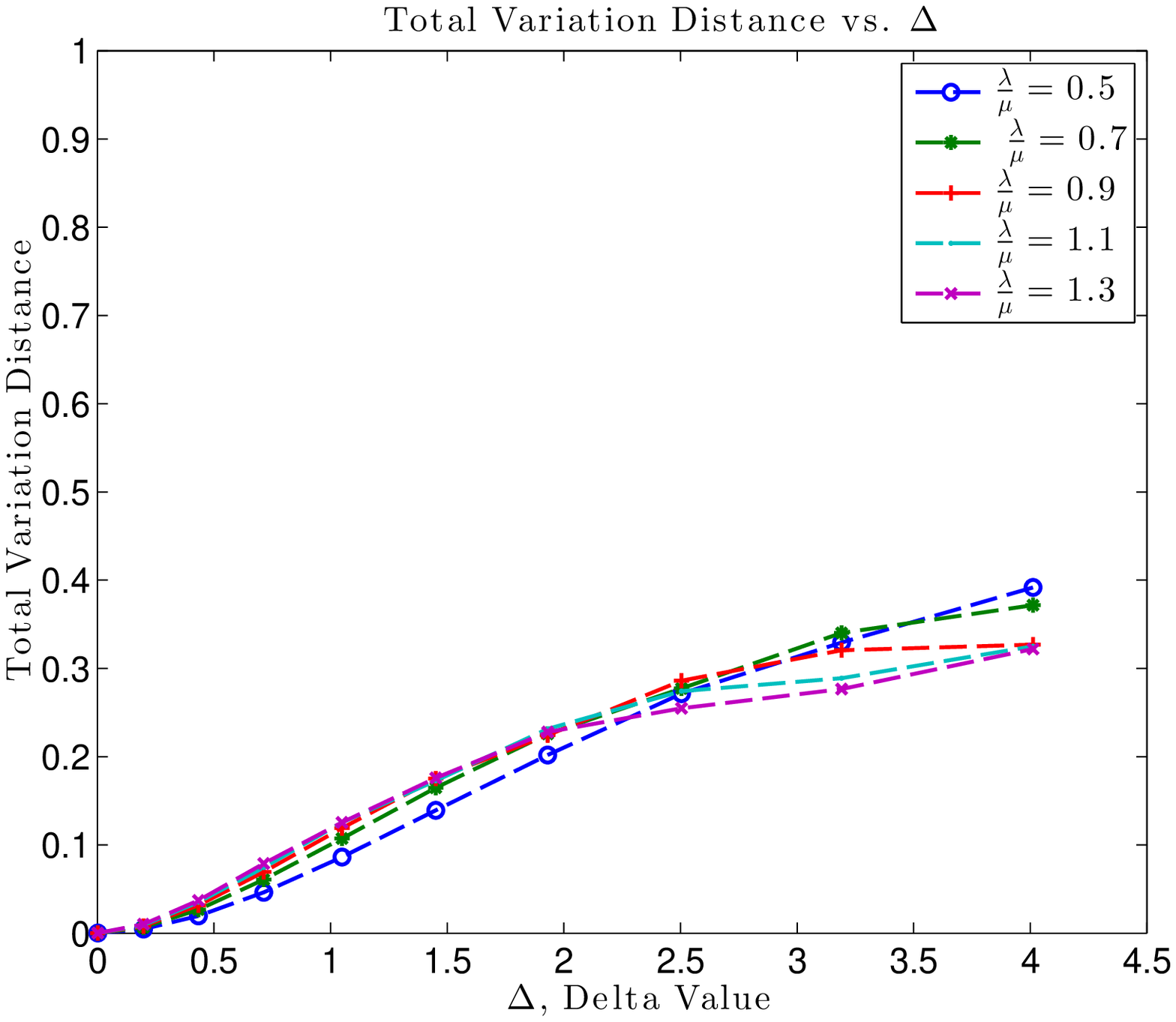}
                \caption{Network A ($\Delta_u = 1$)}
                \label{fig:16PathTVD}
        \end{subfigure}%
        ~ 
        \begin{subfigure}[b]{0.4\textwidth}
                \centering
                \includegraphics[width=\textwidth]{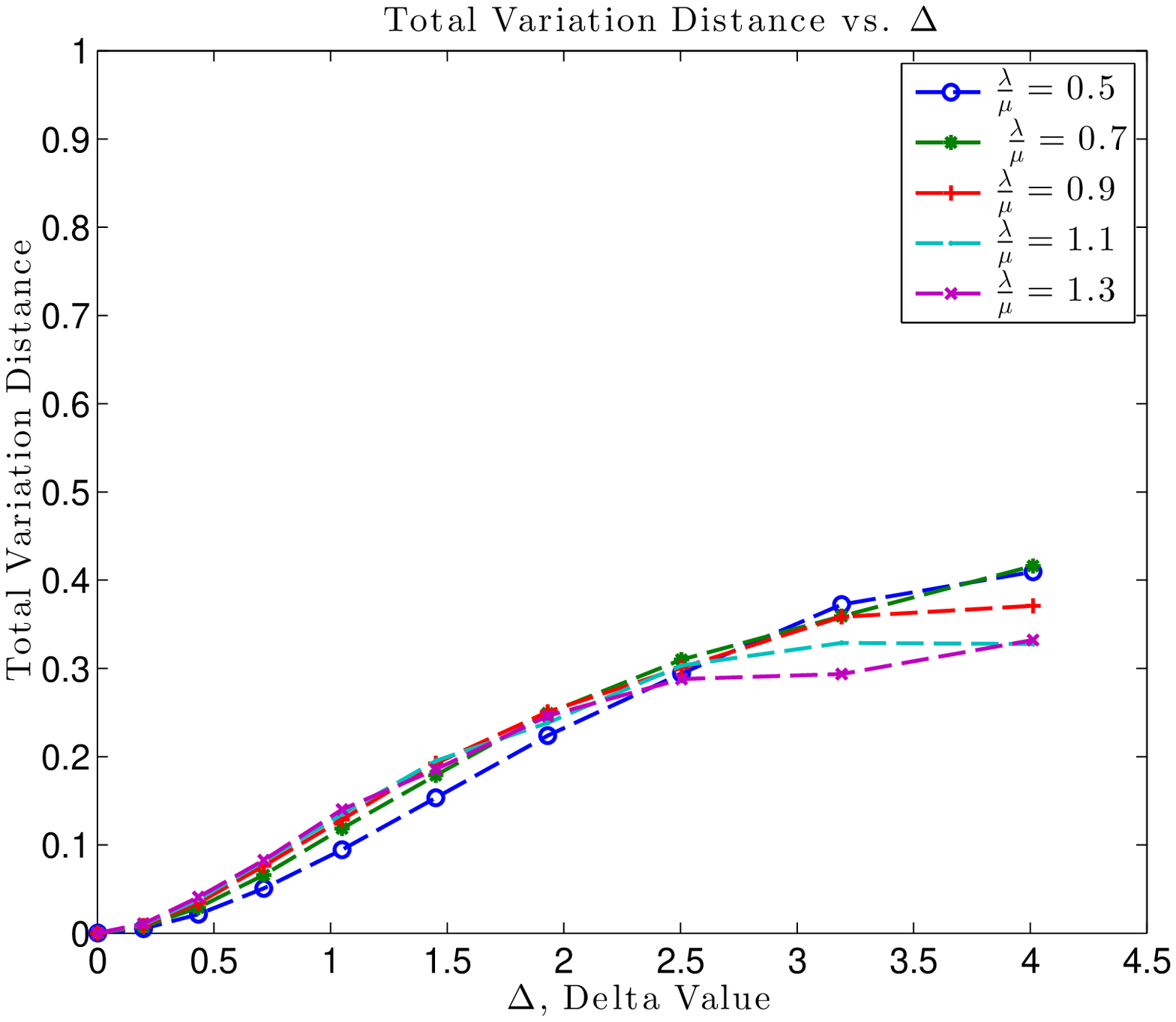}
                \caption{Network B ($\Delta_u = 1$)}
                \label{fig:16CycleTVD}
        \end{subfigure}%

        \begin{subfigure}[b]{0.4\textwidth}
                \centering
                \includegraphics[width=\textwidth]{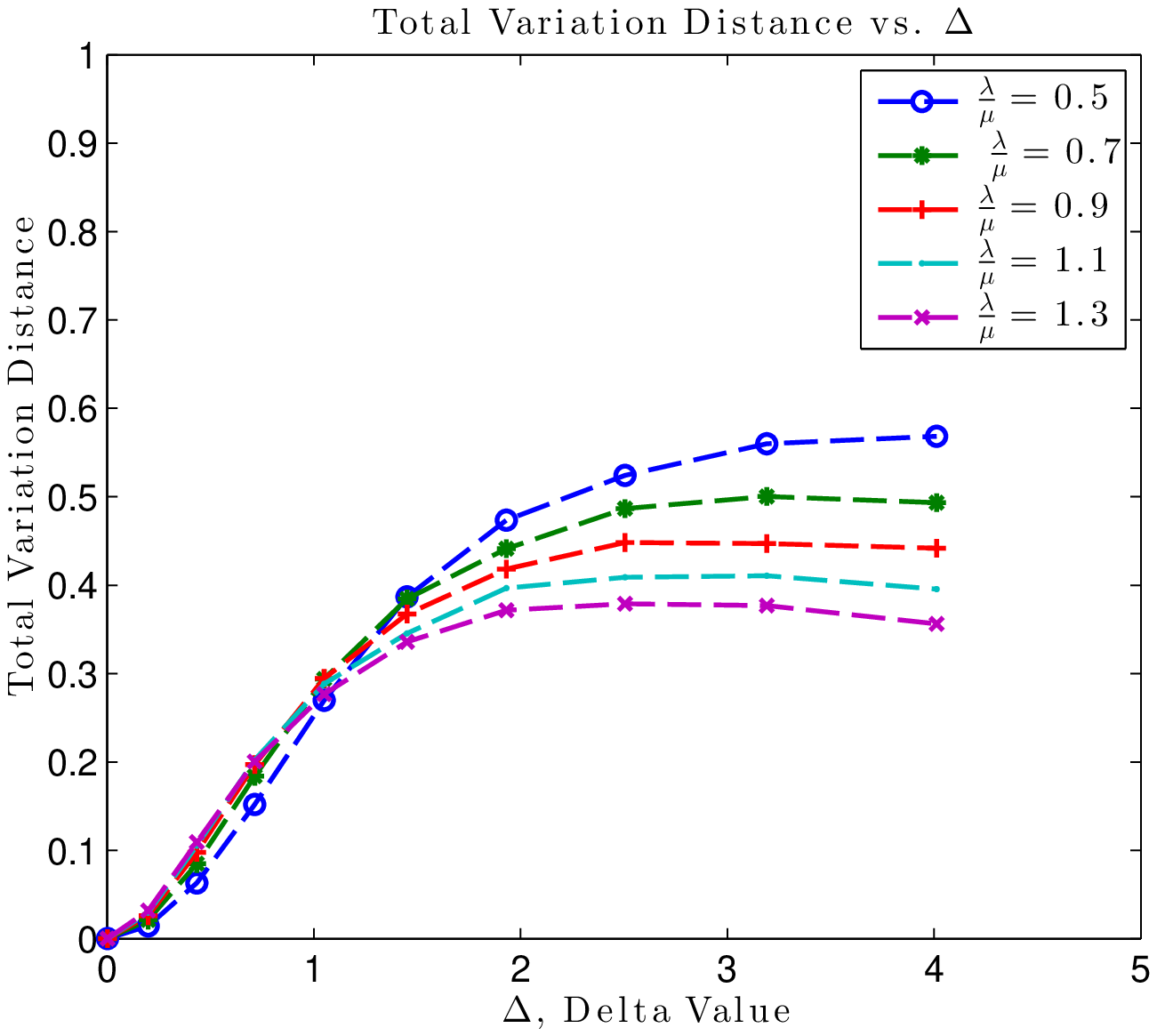}
                \caption{Network C ($\Delta_u = 0.32$)}
                \label{fig:16ERaTVD}
        \end{subfigure}%
            ~ 
         \begin{subfigure}[b]{0.4\textwidth}
                \centering
                \includegraphics[width=\textwidth]{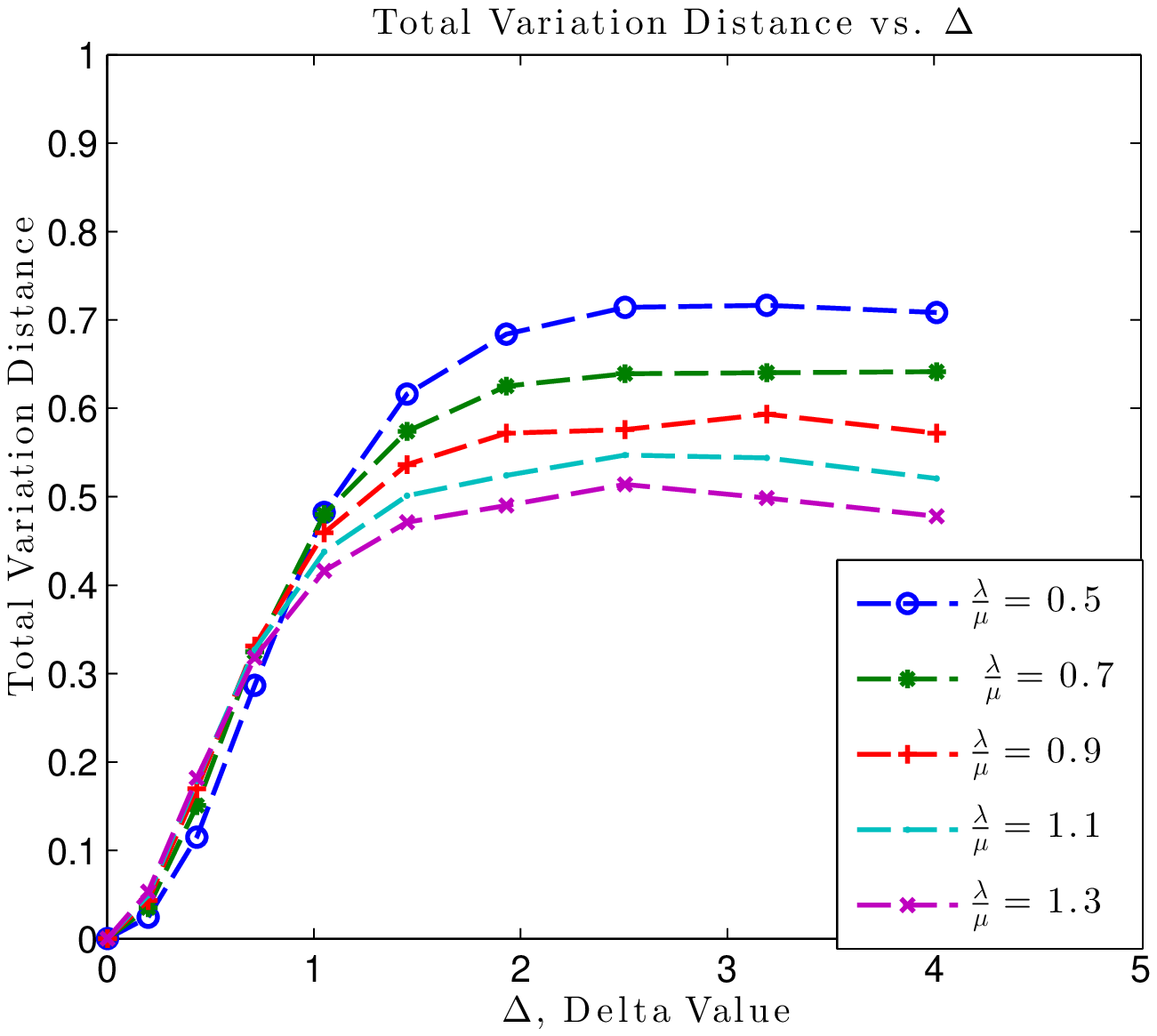}
                \caption{Network D ($\Delta_u = 0.32$)}
                \label{fig:16ERbTVD}
        \end{subfigure}%
       
        \begin{subfigure}[b]{0.4\textwidth}
                \centering
                \includegraphics[width=\textwidth]{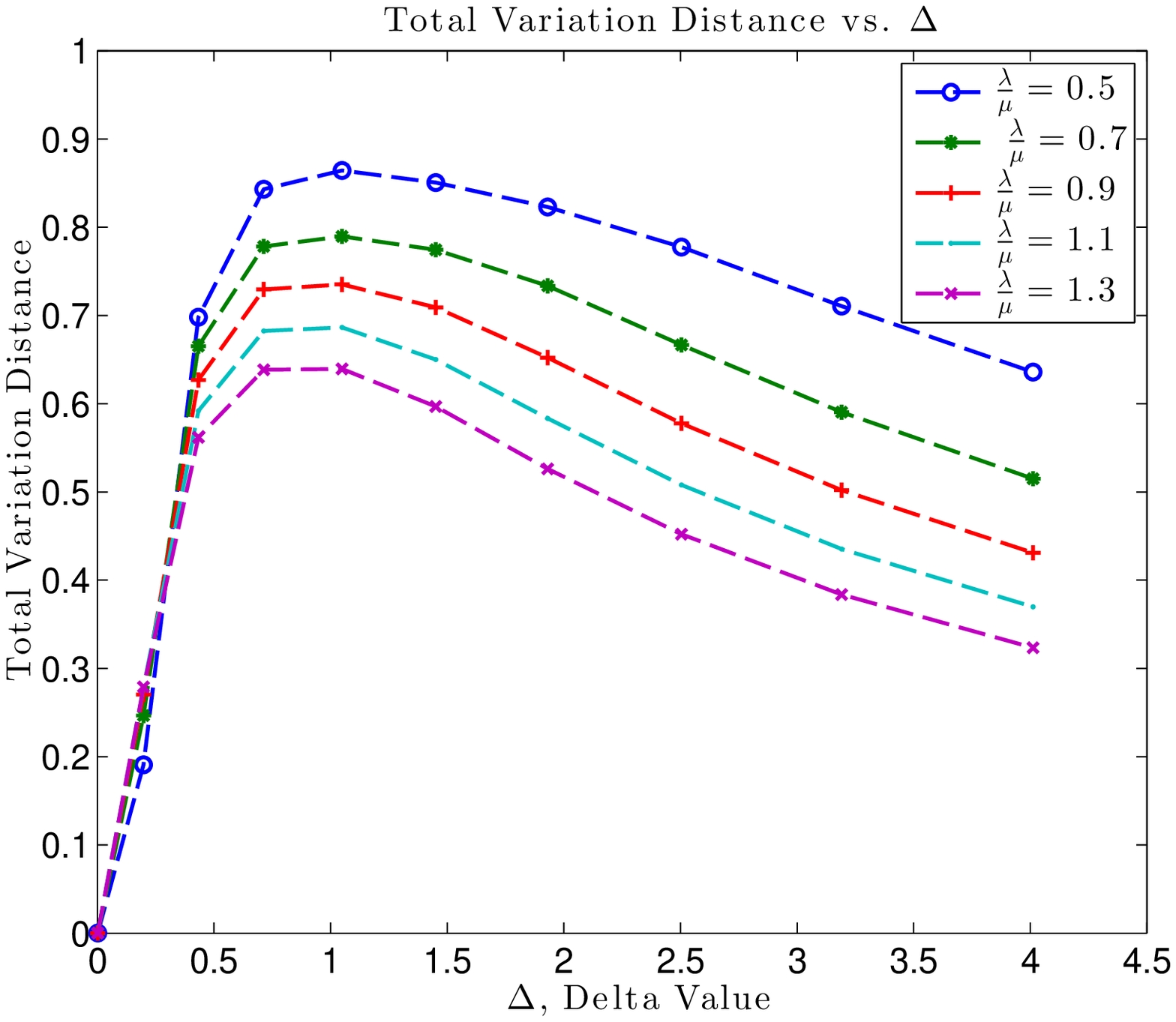}
                \caption{Network E ($\Delta_u = 0.135$)}
                \label{fig:16WSaTVD}
        \end{subfigure}%
         ~ 
        \begin{subfigure}[b]{0.4\textwidth}
                \centering
                \includegraphics[width=\textwidth]{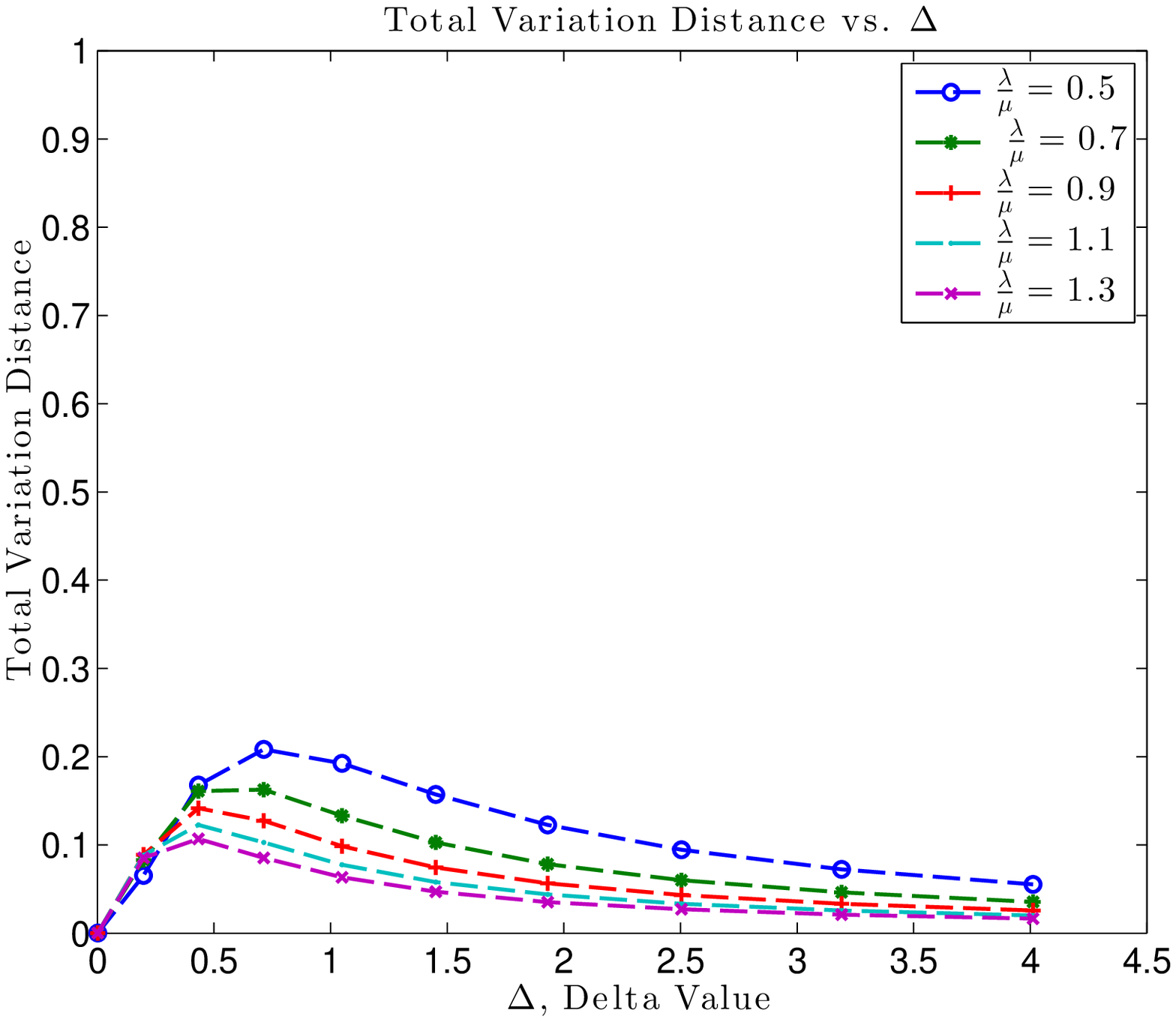}
                \caption{Network F ($\Delta_u = 0.098$)}
                \label{fig:16StarTVD}
        \end{subfigure}%

	\caption{Dependence of $\text{TVD}(\pi_{e}, \pi_{\scriptsize\textrm{approx}})$ on $\Delta$}\label{fig:networkexample1TVD}
 \end{figure}

\section{Most-Probable Configuration}\label{sec:vulnerablesub}

We showed in Figures~\ref{fig:eqvsplotsinside} and~\ref{fig:eqvsplotsoutside} that, for a range of the dynamic parameters, the equilibrium distribution $\pi_e(\mathbf{x})$ of the extended contact process is well approximated by the equilibrium distribution $\pi_{\scriptsize\textrm{approx}}(\mathbf{x})$ of the scaled SIS process. In this section, we consider the problem of finding the most-probable configuration (i.e., configuration with the maximum equilibrium probability). 

For network processes, the most-probable configuration depends on the infection and healing rates and on the underlying network topology. It identifies the set of agents that are most likely to be infected in the long run. These are the more vulnerable agents in the network. If the most probable configuration is $\mathbf{x}^0 = [0,0, \ldots 0]^T$, all agents are healthy, whereas if the most-probable configuration is $\mathbf{x}^N = [1, 1,\ldots 1]^T$, then all agents are at risk regardless of their location in the network. Except for these two cases, finding which agents are infected in the most-probable configuration is not trivial.  
The most-probable configuration of the extended contact process is
\[
\mathbf{x}^*_e = \arg \max_{\mathbf{x} \in \mathcal{X}} \pi_e(\mathbf{x}),
\]
where $\mathcal{X}$ is the set of all $2^N$ possible network configurations. For the extended contact process, there is no closed-form description of the equilibrium distribution, so, this problem can only be solved numerically, which is infeasible for large-scale networks. 

On the other hand, as stated in Theorem~\ref{theorem:scaledcontactapprox}, when $\Delta << \Delta_u$, the equilibrium distribution, $\pi_e(\mathbf{x})$, of the extended contact process is well approximated by the equilibrium distribution, $\pi_{\scriptsize\textrm{approx}}(\mathbf{x})$, of a scaled SIS process with endogenous infection rate $\beta_s = 1+ \Delta$. We proved in \cite{JZhangJournal2} that, in this case, the most-probable configuration of the scaled SIS process can be solved in \emph{polynomial-time} because it corresponds to solving for the minimum of a submodular function. It is therefore possible to identify vulnerable network substructures for networks with hundreds and thousands of agents.


From the simulation results in the previous section, we now compare the most-probable configuration of the extended contact process with the most-probable configuration of the approximating scaled SIS process. Table~\ref{table:mpc2} lists for the six networks in Figure~\ref{fig:networkexample}, the TVD between the distributions, the corresponding most-probable configurations, and the probabilities of the most-probable configuration for $\frac{\lambda}{\mu} = 0.9744$ and $\Delta = 0.02$. We observe that when the condition of Theorem~\ref{theorem:scaledcontactapprox} is satisfied: 
\begin{enumerate}
\item the most-probable configuration, $\mathbf{x}^*_e$, of the extended contact process is the approximately the same as the most-probable configuration, $\mathbf{x}^*_{\scriptsize\textrm{approx}}$, of the scaled SIS process;
\item the probability of the most-probable configuration, $\pi_e(\mathbf{x}^*_e)$, of the extended contact process is the same as the probability of the most-probable configuration, $\pi_{\scriptsize\textrm{approx}}(\mathbf{x}^*_{\scriptsize\textrm{approx}})$, of the scaled SIS process.
\end{enumerate}

%
%
%
%
%
%
%

 \begin{table}[h]
\begin{tabular}{|c | c | c | c | c | c | }
\hline
& TVD$(\pi_e,\pi_{\scriptsize\textrm{approx}})$ &  $\mathbf{x}^*_e$& $\mathbf{x}^*_{\scriptsize\textrm{approx}}$ & $\pi_e(\mathbf{x}^*_{e})$& $\pi_{\scriptsize\textrm{approx}}(\mathbf{x}^*_{\scriptsize\textrm{approx}})$ \\ \hline

Network A & $1.0236 \times 10^{-4}$ & $\mathbf{x}^0 = [0, 0, \ldots 0]^T$ & $\mathbf{x}^0= [0, 0, \ldots 0]^T $&  $1.7431 \times 10^{-5}$ & $1.7427\times 10^{-5}$ \\ \hline

Network B & $1.1027 \times 10^{-4}$ & $\mathbf{x}^0 = [0, 0, \ldots 0]^T$ & $\mathbf{x}^0= [0, 0, \ldots 0]^T $&  $1.7347 \times 10^{-5}$ & $1.7342\times 10^{-5}$ \\ \hline

Network C & $3.3806 \times 10^{-4}$ & see Figure~\ref{fig:mpcnondegen} & see Figure~\ref{fig:mpcnondegen} &  $1.7107 \times 10^{-5}$ & $1.7154\times 10^{-5}$ \\ \hline

Network D & $5.2714 \times 10^{-4}$ & $\mathbf{x}^N = [1, 1, \ldots 1]^T$ & $\mathbf{x}^N = [1, 1, \ldots 1]^T$&  $1.8622 \times 10^{-5}$ & $1.8781\times 10^{-5}$ \\ \hline

Network E & $0.0031$ & $\mathbf{x}^N = [1, 1, \ldots 1]^T$ & $\mathbf{x}^N = [1, 1, \ldots 1]^T$&  $3.1073 \times 10^{-5}$ & $3.277\times 10^{-5}$ \\ \hline

Network F & $0.0023$ & $\mathbf{x}^0 = [0, 0, \ldots 0]^T$ & $\mathbf{x}^N = [0, 0, \ldots 0]^T$&  $1.7419 \times 10^{-5}$ & $1.7389\times 10^{-5}$ \\ \hline
\end{tabular}\caption{Most-Probable Configuration when $\frac{\lambda}{\mu} = 0.9744$ and $\Delta = 0.02$.}\label{table:mpc2}
\end{table}


%
 
 \begin{figure}[ht]      
      \centering
          \begin{subfigure}[b]{0.25\textwidth}
                \centering
                \includegraphics[width=\textwidth]{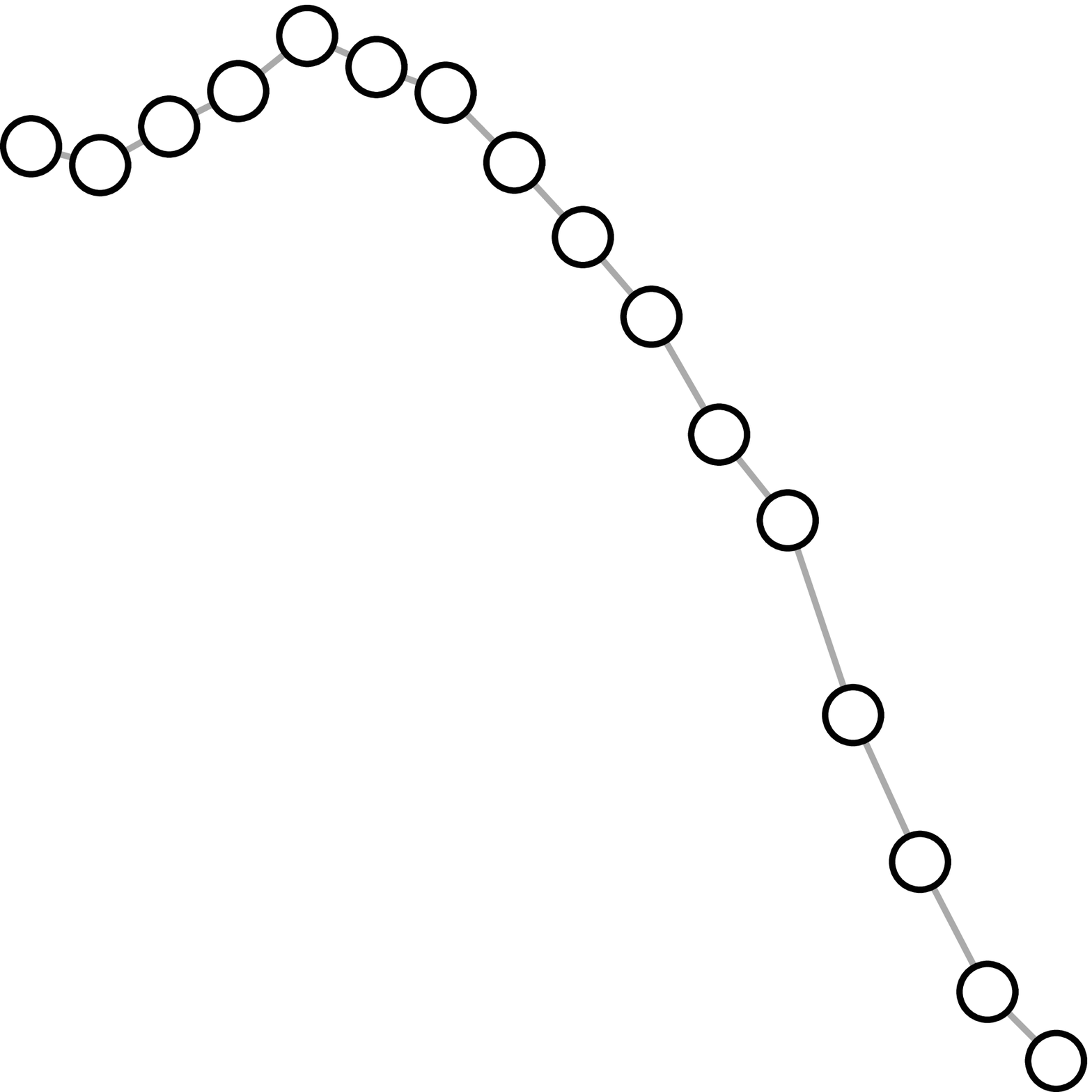}
                \caption{Network A $(d(G) = 0.9375)$}
                \label{fig:16Pathmpc}
        \end{subfigure}
	  \enspace
        \begin{subfigure}[b]{0.25\textwidth}
                \centering
                \includegraphics[width=\textwidth]{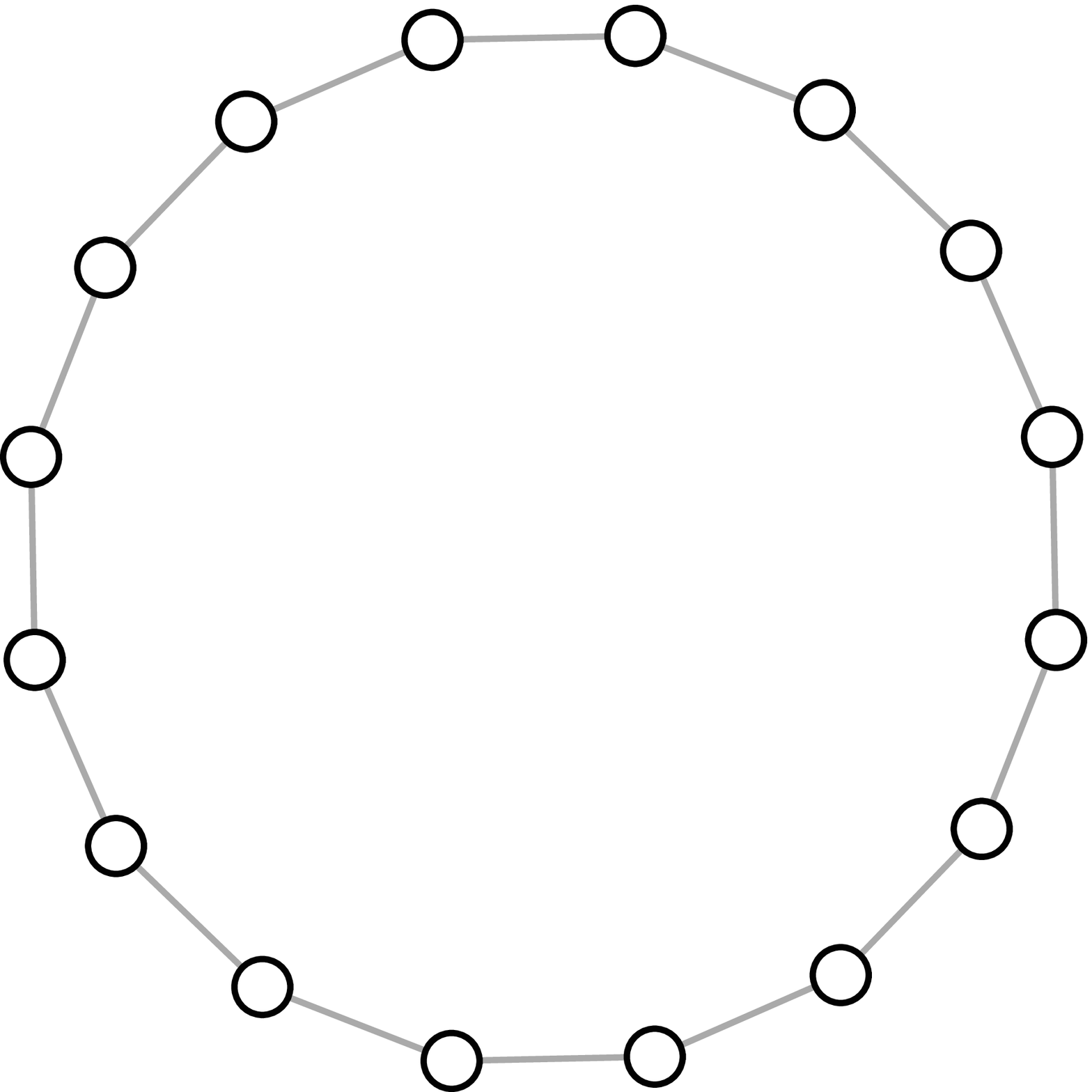}
                \caption{Network B\\ $(d(G) = 1)$}
                \label{fig:16Cyclempc}
        \end{subfigure}%
         \enspace
          \begin{subfigure}[b]{0.25\textwidth}
                \centering
                \includegraphics[width=\textwidth]{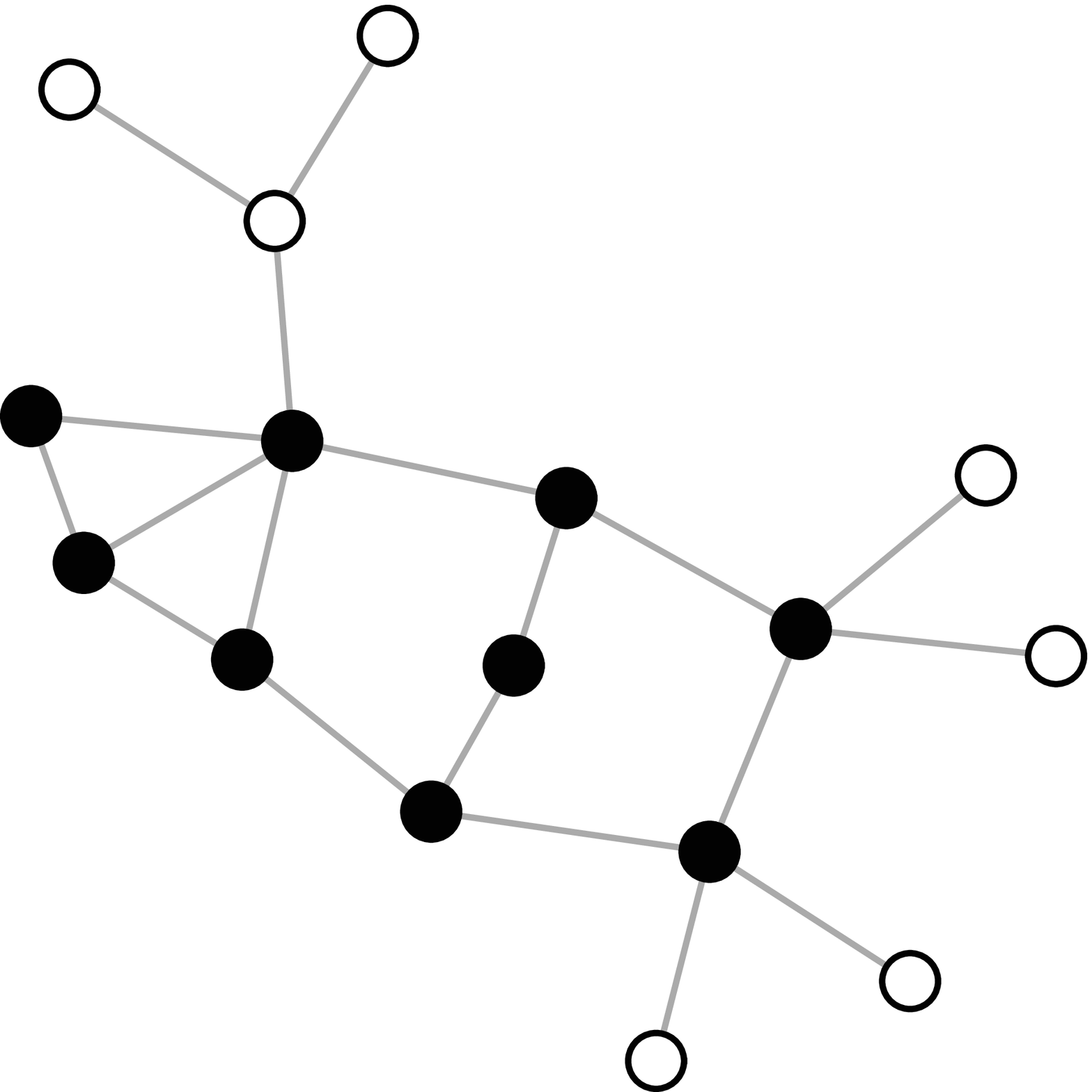}
                \caption{Network C $(d(G) = 1.1875)$}
                \label{fig:mpcnondegen}
        \end{subfigure}
        
        \hfill \\
          \begin{subfigure}[b]{0.25\textwidth}
                \centering
                \includegraphics[width=\textwidth]{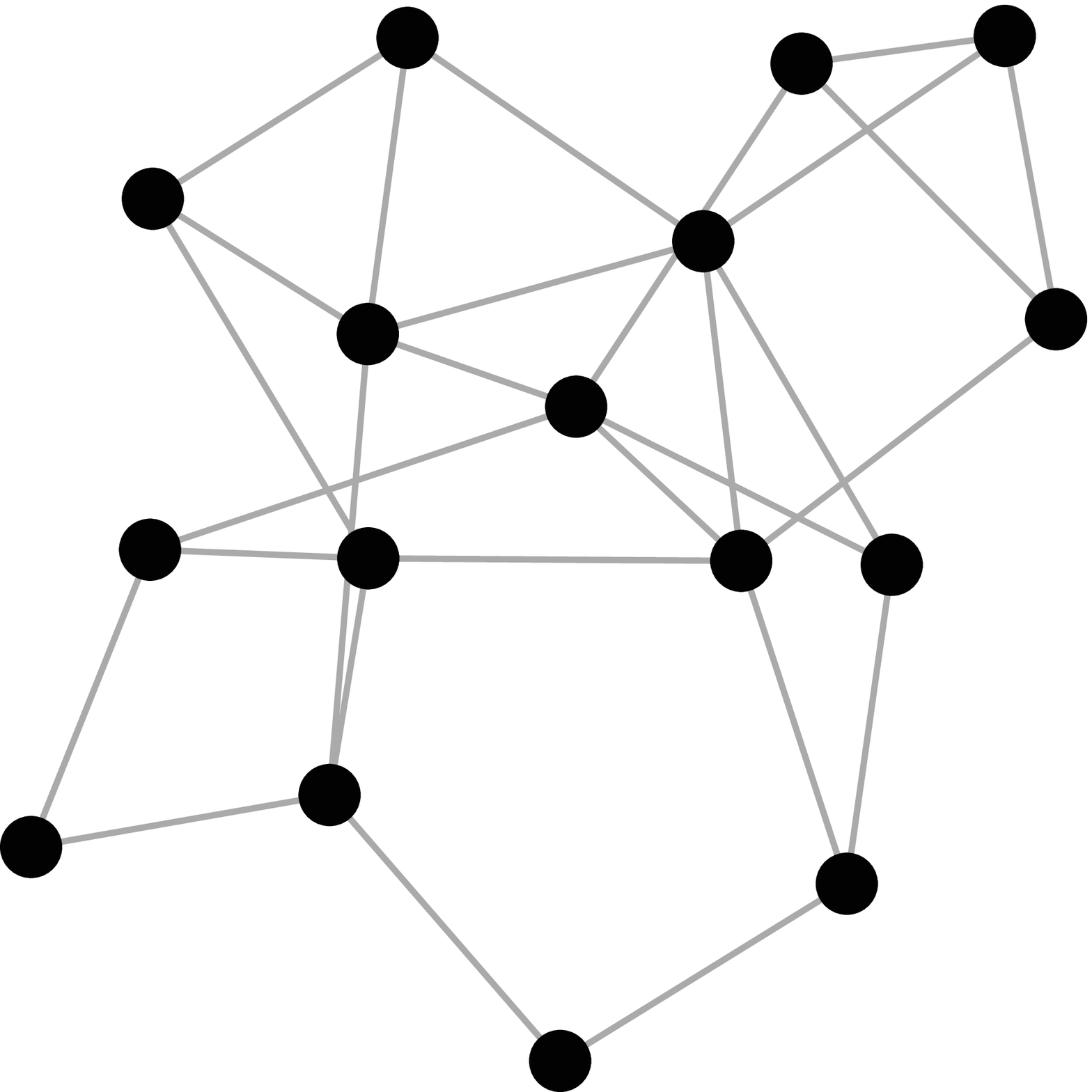}
                \caption{Network D $(d(G) = 1.75)$}
                \label{fig:16ERbmpc}
        \end{subfigure}
	 \enspace 
         \begin{subfigure}[b]{0.25\textwidth}
                \centering
                \includegraphics[width=\textwidth]{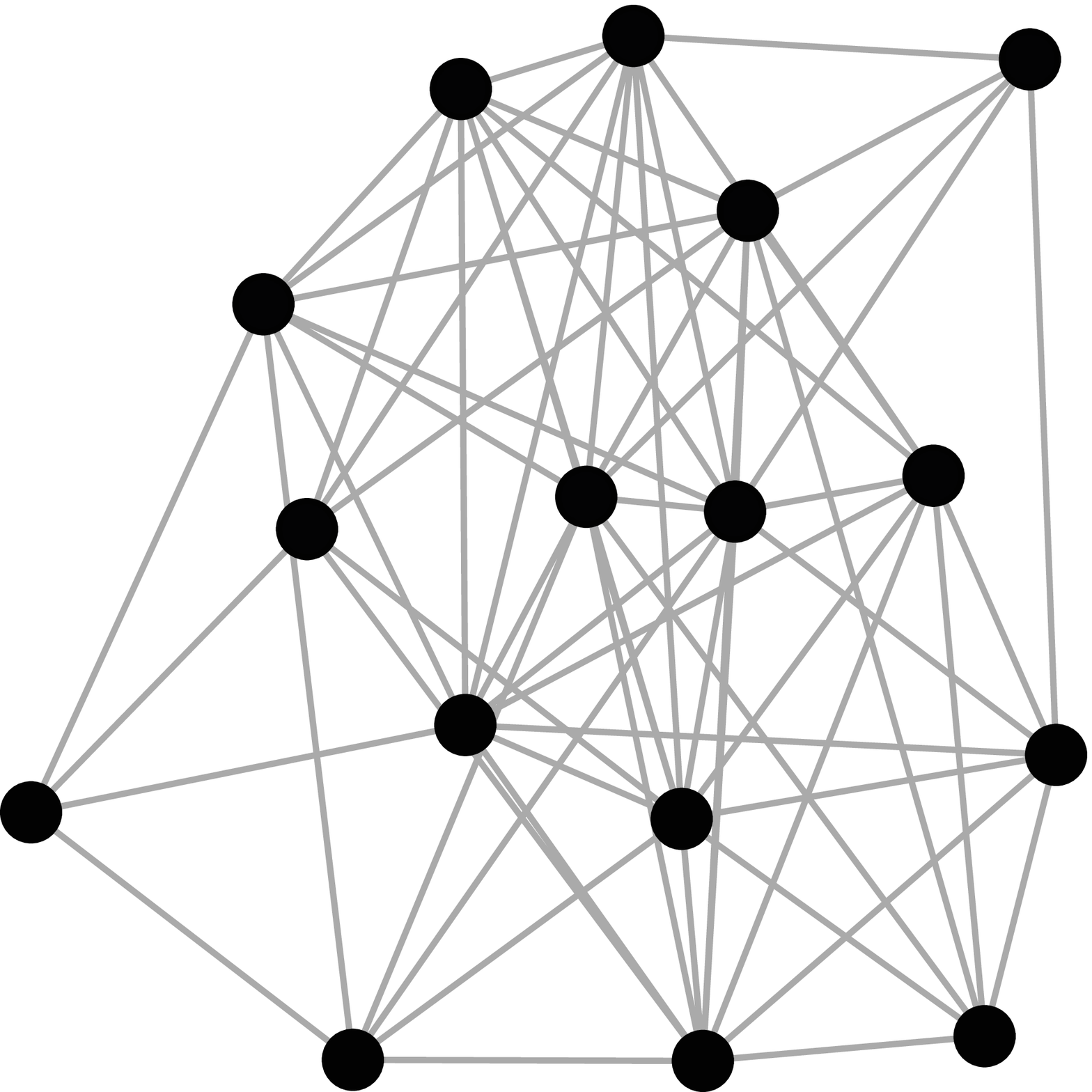}
                \caption{Network E $(d(G) = 4.125)$}
                \label{fig:16WSampc}
        \end{subfigure}
       \enspace 
         \begin{subfigure}[b]{0.25\textwidth}
                \centering
                \includegraphics[width=\textwidth]{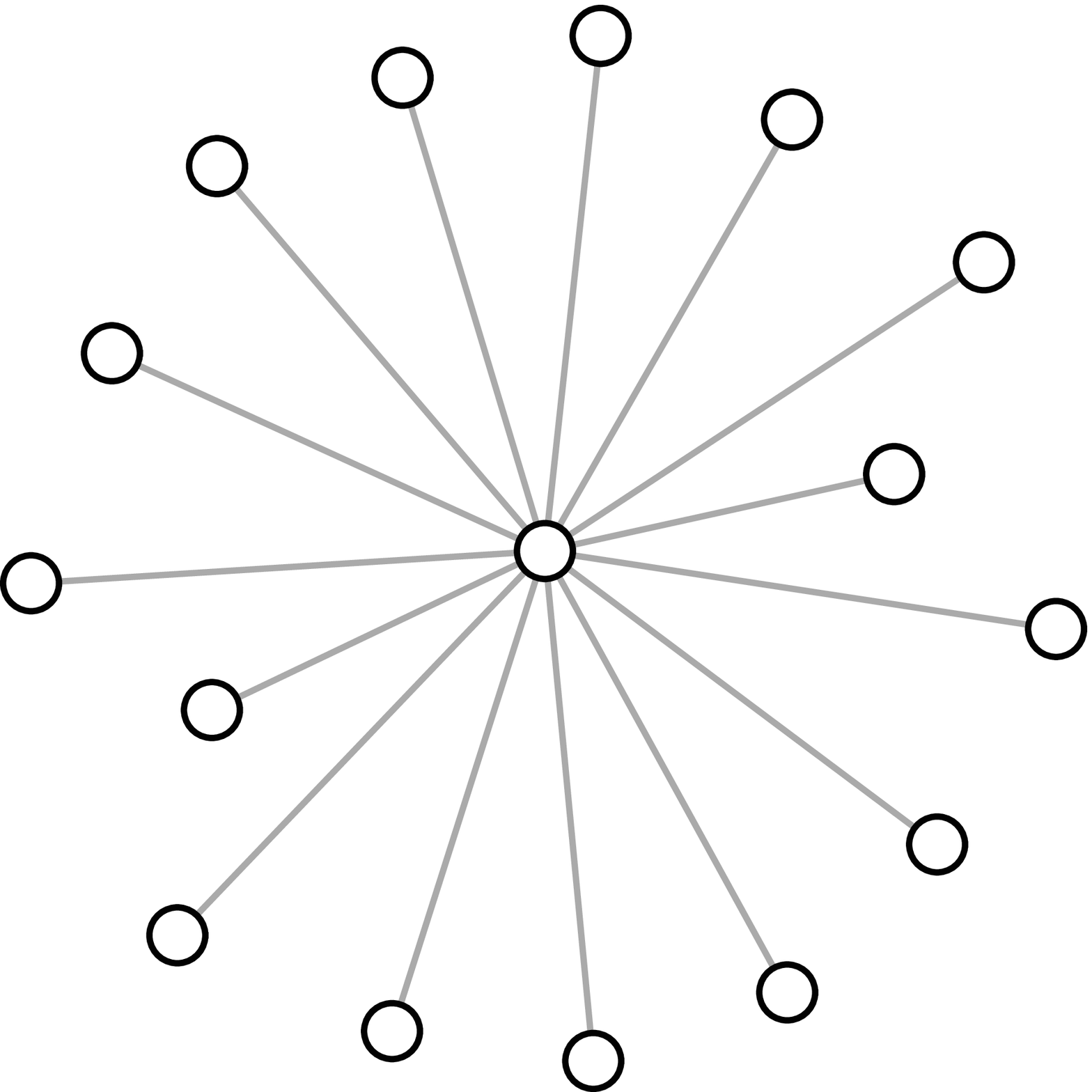}
                \caption{Network F $(d(G) = 0.9375)$}
                \label{fig:16Starmpc}
        \end{subfigure}

	\caption{Most Probable Configuration when $\frac{\lambda}{\mu} = 0.9744$ and $\Delta = 0.02$ \\(infected = black, healthy = white)}\label{fig:npc}
 \end{figure}

For Networks A, B, and F, the most-probable configuration for both the extended contact process and the approximate scaled SIS process is $\mathbf{x}^0$, the configuration where all the agents are healthy. However, for the same infection and healing rate, the most-probable configuration for Networks D and E for both the extended contact process and the scaled SIS process is $\mathbf{x}^N$, the configuration where all the agents are infected. Figure~\ref{fig:mpcnondegen} shows that the most-probable configuration for Network C is neither $\mathbf{x}^0$ nor $\mathbf{x}^N$ but a configuration where nine agents are infected while seven agents are healthy; we call most-probable configurations that are neither $\mathbf{x}^0$ nor $\mathbf{x}^N$ \emph{non-degenerate} most-probable configurations.

For an extended contact process with exogenous infection rate and healing rate, $\frac{\lambda}{\mu} = 0.9744$, and endogenous infection rate, $\beta_c =\frac{\lambda}{\mu}\Delta = 0.019488$, the epidemic is minor in Networks A, B, and F, but should be of concern in Networks D and E. In Network C, a subset of agents are more at risk than others. Different networks have different risk levels because the propagation of contagious infection is dependent on the underlying network topology. The result in Figure~\ref{fig:mpcnondegen} confirms for the extended contact process what we proved for the scaled SIS process in \cite{JZhangJournal2}, namely, that in the most-probable configuration the infected agents belong to dense subgraphs in the network. Reference \cite{Khuller:2009vj} defines density of a graph $G$ by 
\[
d(G) = \frac{\left | E(G) \right |}{ \left | V(G) \right |},
\]
where $\left | E(G) \right |$ is the total number of edges and  $\left | V(G) \right |$ is the total number of nodes. Networks that are more connected have higher densities than sparsely connected networks.  

Networks with high density, such as Networks D and E are more at risk to contagion than networks with low density such as Networks A, B, and F. Network F, although it has the largest maximum degree, has the same density as Network A. It is difficult for infection to spread in Network F because the center agent is the only agent capable of transmitting the infection to its neighbors. We showed in \cite{JZhangJournal2} that the nine infected agents in Network C are more at risk of infection than the other agents because they form a subgraph that is denser than the overall network; these nine agents are especially well-connected in this network.

 \begin{table}[h]
\begin{tabular}{|c | c | c | c | c | c | }
\hline
& TVD$(\pi_e,\pi_{\scriptsize\textrm{approx}})$ &  $\mathbf{x}^*_e$& $\mathbf{x}^*_{\scriptsize\textrm{approx}}$ & $\pi_e(\mathbf{x}^*_{e})$& $\pi_{\scriptsize\textrm{approx}}(\mathbf{x}^*_{\scriptsize\textrm{approx}})$ \\ \hline

Network A & $0.0266$ & $\mathbf{x}^0 = [0, 0, \ldots 0]^T$ & $\mathbf{x}^0= [0, 0, \ldots 0]^T $&  $6.7989 \times 10^{-5}$ & $6.4085\times 10^{-5}$ \\ \hline

\cellcolor[gray]{0.8}Network B & \cellcolor[gray]{0.8}$0.029$ & \cellcolor[gray]{0.8}$\mathbf{x}^0 = [0, 0, \ldots 0]^T$ & \cellcolor[gray]{0.8}$\mathbf{x}^N= [1, 1, \ldots 1]^T $&  \cellcolor[gray]{0.8}$6.2942 \times 10^{-5}$ & \cellcolor[gray]{0.8}$6.1972\times 10^{-5}$ \\ \hline

\cellcolor[gray]{0.8}Network C & \cellcolor[gray]{0.8}$0.0848$ & \cellcolor[gray]{0.8}see Figure~\ref{fig:mpcnondegen} & \cellcolor[gray]{0.8}$\mathbf{x}^N= [1, 1, \ldots 1]^T $ &  \cellcolor[gray]{0.8}$7.0847 \times 10^{-5}$ & \cellcolor[gray]{0.8}$1.214\times 10^{-4}$ \\ \hline

Network D & $0.1505$ & $\mathbf{x}^N = [1, 1, \ldots 1]^T$ & $\mathbf{x}^N = [1, 1, \ldots 1]^T$&  $2.5957\times 10^{-4}$ & $0.0011$ \\ \hline

Network E & $0.6652$ & $\mathbf{x}^N = [1, 1, \ldots 1]^T$ & $\mathbf{x}^N = [1, 1, \ldots 1]^T$&  $0.0066$ & $0.1849$ \\ \hline

Network F & $0.1609$ & $\mathbf{x}^0 = [0, 0, \ldots 0]^T$ & $\mathbf{x}^N = [0, 0, \ldots 0]^T$&  $5.988 \times 10^{-5}$ & $3.7915\times 10^{-5}$ \\ \hline
\end{tabular}\caption{Most-Probable Configuration when $\frac{\lambda}{\mu} = 0.7$ and $\Delta = 0.4333$.}\label{table:mpc3}
\end{table}


Table~\ref{table:mpc3} lists the TVD between the distributions, the most-probable configurations for the extended contact process and the approximate scaled SIS process, and the probabilities of the most-probable configurations for $\frac{\lambda}{\mu} = 0.7$ and $\Delta = 0.4333$; the factor, $\Delta$, no longer satisfies the condition of Theorem~\ref{theorem:scaledcontactapprox}. As a result, the TVD between the distributions are larger and for Networks B and C, the most-probable configurations of the extended contact process and the approximate scaled SIS process are no longer the same. One intuition as to why for Network C,  $\mathbf{x}^*_{\scriptsize\textrm{approx}}=\mathbf{x}^N$ (the configuration where all the agents are infected) while $\mathbf{x}^*_e = \mathbf{x}^0$ (the configuration where all the agents are healthy) is because the endogenous infection rate of the scaled SIS is exponentially dependent on the number of infected neighbors while the endogenous infection rate of the extended contact process is linearly dependent on the number of infected neighbors; contagion is more virulent because the endogenous infection rate is higher in the scaled SIS process. 

Note that the configuration in Figure~\ref{fig:mpcnondegen}, where nine agents are more at risk of infection than others, remains the most-probable configuration for Network C. Even though this configuration no longer has the highest equilibrium probability in the approximate distribution, it remains a highly probable configuration. This reinforces our observation from Figure~\ref{fig:eqvsplotsoutside} that configurations with high probabilities in the approximate distribution are also highly probable in the equilibrium distribution of the extended contact process. The substructures that are vulnerable for the scaled SIS process, the \emph{non-degenerate} most-probable configurations, are also vulnerable substructures of the extended contact process.

\section{Conclusion}\label{sec:conclusion}
This paper considers conditions under which the extended contact process \cite{PhysRevE.86.016116, JZhang} is well approximated by a scaled SIS process \cite{JZhangJournal}. This is important because the extended contact process can only be studied by numerical means, except in the trivial complete network, whereas for the scaled SIS process, we have a closed-form solution for its equilibrium distribution. Both processes model Markov dynamic processes over a static, undirected network. The extended contact process, also called the $\epsilon$-SIS process, is a modification of the basic contact process  \cite{liggett1999stochastic} to include a nonzero exogenous infection rate. The contact process is often used as model of diffusion of virus or information over networks. 

We are interested in understanding how the time-asymptotic behavior of dynamical network processes in particular, of the extended contact process, depends on the underlying network topology. The equilibrium distribution of the extended contact process, although well-defined, requires solving numerically an eigenvalue-eigenvector problem for a $2^N \times 2^N$ matrix, which is infeasible except for very small size networks. It is also not analytically available for arbitrary network topologies. In this paper, we show how the equilibrium distribution of the scaled SIS process, which does have a closed-form analytical description \cite{JZhangJournal}, is, under a certain range of parameter values, a good approximation to the equilibrium distribution of the extended contact process.

The extended contact process assumes that the infection rate of a susceptible agent has a linear dependence on the number of infected neighbors, whereas the scaled SIS process assumes that the infection rate is exponentially dependent on the number of infected neighbors. The paper gives a conservative upperbound on the endogenous infection rate, $\beta_e$, of the extended contact process for which the equilibrium distribution of the extended contact process, $\pi_e(\mathbf{x})$ is well approximated by that of an equivalent scaled SIS process, $\pi_{\scriptsize\textrm{approx}}(\mathbf{x})$. We showed that this upperbound depends on the maximum degree, $d_{\max}$, of the underlying network topology.

We confirmed these results with simulations using six different networks with 16 nodes (for which can solve numerically the associated eigenvector-eigenvalue problem). We compared the true equilibrium distribution, $\pi_e(\mathbf{x})$, of the extended contact process, obtained numerically, with that of the approximate distribution, $\pi_{\scriptsize\textrm{approx}}(\mathbf{x})$, derived from the scaled SIS process. By checking the deviation for parameters both within and outside the theoretical bound, we confirmed that the proposed approximation is good. Depending on the underlying network, the total variation distance (TVD) between the true equilibrium distribution and its approximation may still remain small even for processes where the parameter values are larger than the upperbound. Further, the TVD does not necessarily increase for increasing deviation from the theoretical upperbound. In future work, we would like to explore how the structure of the transition rate matrices leads to a decrease in deviation between the extended contact process and the approximation.

We then used this approximation to determine the most-probable configuration of the extended contact process. Unlike the scaled SIS process, we do not have any bounds for the extended contact process on the infection and healing rate as to when the most-probable configuration is $\mathbf{x}^0$, $\mathbf{x}^N$, or a non-degenerate configuration. When the TVD is small, the configurations with the highest equilibrium probability are identical for both the extended contact process and the scaled SIS process. This means that we can use the scaled SIS process, whose most-probable configuration can be found in polynomial-time, to find the most-probable configuration of the extended contact process. The most-probable configurations reveal subgraphs in the network that are more vulnerable to infection by the extended contact process.


\section*{Acknowledgement}
We thank Microsoft for providing us their cloud computational resources with a Microsoft Azure Research Award. 

\bibliographystyle{IEEEtran}
\bibliography{refs}


\appendices
\section{Proof of Lemma~\ref{lemma:approx}}\label{proof:lemma:approx}
\begin{lemma}
For any nonnegative integer $m$ from $0$ to $d_{\max}$, if 
\[
\Delta^2 << \frac{2}{d_{\max}(d_{\max} - 1)},
\]
then, for $\beta = 1+\Delta$,
\[
\frac{\lambda}{\mu} \beta^m  = \frac{\lambda}{\mu}(1+\Delta)^{m} \approx  \frac{\lambda}{\mu} + \frac{\lambda}{\mu}\Delta m.
\]
\end{lemma}

\begin{proof}
From the binomial series, for integer $m \in \{0, 1, \ldots d_{\max}\}$,
\begin{align*}
\frac{\lambda}{\mu}\beta^m &= \frac{\lambda}{\mu}(1+\Delta)^m =  \frac{\lambda}{\mu}\left(\sum_{k=0}^m {m \choose k}\Delta^k\right) \\
& = \frac{\lambda}{\mu}\left({m \choose 0}\Delta^0 + {m \choose 1}\Delta + {m \choose 2}\Delta^2 + {m \choose 3}\Delta^3 \ldots + {m \choose m}\Delta^m\right)\\
& =\frac{\lambda}{\mu}\left(1 + m\Delta+ \frac{m(m-1)}{2}\Delta^2 + \ldots m\Delta^{m-1} +\Delta^m \right).
\end{align*}

If 
\begin{equation}\label{eq:negligible}
\left\{{m \choose 2}\Delta^2, {m \choose 3}\Delta^3, \ldots {m \choose \frac{m}{2}}\Delta^{\frac{m}{2}}\right\} << 1, \quad \forall \,m \in \{0, 1, \ldots d_{\max}\}, 
\end{equation}
then the quadratic and higher order terms in the summation are negligible and we obtain the linear approximation
\[
\frac{\lambda}{\mu}\beta^{m} \approx  \frac{\lambda}{\mu} + \frac{\lambda}{\mu}\Delta m,
\] 
which holds for all $m$.

%
Recognize that for $m \in \{0, 1, \ldots d_{\max}\}$, $m > \frac{m-1}{2} > \frac{m-2}{3} > \ldots > \frac{\frac{m}{2}-1}{\frac{m}{2}}$. This means that
\[
{m \choose 2}\Delta^2 > {m \choose 3}\Delta^3 > \ldots > {m \choose \frac{m}{2}}\Delta^{\frac{m}{2}}, \quad \forall \,m \in \{0, 1, \ldots d_{\max}\}. 
\]
The largest possible upperbond is when $m = d_{\max}$. Therefore, condition \eqref{eq:negligible} is satisfied when $\frac{d_{\max}(d_{\max}-1)}{2}\Delta^2 << 1$.

\end{proof}

\section{Proof of Theorem~\ref{theorem:scaledcontactapprox}}\label{proof:theorem:scaledcontactapprox}

\begin{theorem}
Consider the extended contact process exogenous infection rate $\lambda$, healing rate $\mu$, and endogenous infection rate $\beta_e$, over a static, simple, connected, undirected network of arbitrary topology, $G$, with maximum degree $d_{\max}$. Let $\beta_e = \frac{\lambda}{\mu}\Delta$. If 
\[
\Delta^2 << \frac{2}{d_{\max}(d_{\max} - 1)},
\]
then the equilibrium distribution of the extended contact process is well approximated by 
\begin{equation*}
\pi_{\scriptsize\textrm{approx}}(\mathbf{x})= \frac{1}{Z}\left( \frac{\lambda}{\mu}\right)^{1^T{\mathbf x}}  (1 + \Delta)^{\frac{{\mathbf x}^TA{\mathbf x}}{2}  }, \quad  \mathbf{x} \in \mathcal{X} ,
\end{equation*}
where $A$ is the adjacency matrix of the network $G$, and $Z$ is the partition function. The approximate distribution, $\pi_{\scriptsize\textrm{approx}}(\mathbf{x})$, is the equilibrium distribution \eqref{eq:eqapprox} of an equivalent scaled SIS process over the same network $G$ with exogenous infection rate $\lambda$, healing rate $\mu$, and endogenous infection rate $\beta_s = 1+ \Delta$.
\end{theorem}

\begin{proof}
From the theory of continuous-time Markov processes \cite{norris1998markov}, the equilibrium distribution of the extended contact process is the left eigenvector of the transition rate matrix, $\mathbf{Q}_e$, corresponding to the 0 eigenvalue:
\[
\pi \mathbf{Q}_e =0
\]
Entries of the matrix $\mathbf{Q}_e$ correspond to the transition rates from one configuration $\mathbf{x} \in \mathcal{X}$ to another configuration according to the rates~\eqref{eq:contacthealrate} and~\eqref{eq:extendedcontactinfectrate}. 

Lemma~\ref{lemma:approx} gave the condition for when the infection rates (normalized by the healing rate) of the extended contact process are approximately the same as those of the scaled SIS process. As a result, the transition rate matrix of both processes are approximately the same. Therefore, the left eigenvector corresponding to the 0 eigenvalue of $\mathbf{Q}_e$ is also the left eigenvector corresponding to the 0 eigenvalue of $\mathbf{Q}_s$, the transition rate matrix of the scaled SIS process with entries generated according to~\eqref{eq:scaledhealrate} and~\eqref{eq:scaledinfectrate}. We know that the left eigenvector of interest for the rate matrix, $\mathbf{Q}_s$, of the scaled SIS process is given by the closed-form equation \eqref{eq:equilibriumdistribution}.

%
%
\end{proof}

\section{Properties of Functions of Exponentially Distributed Random Variables\cite{papoulis2002probability}}\label{expmin}

\subsection{Two Independent Random Variables}
Let $A \sim \exp(\alpha)$, $B \sim \exp(\beta)$ be two independent exponentially distributed random variables. Then,
\begin{equation*}
\begin{split}
P(A \le B) &= P(A - B \le 0) \\ 
&= \int_0^\infty \int_0^b \alpha e^{-\alpha x} \beta e^{-\beta b} \mathrm{d}a\,\mathrm{d}b\\
&=  \int_0^\infty \left( \int_0^b \alpha e^{-\alpha x} \mathrm{d}a \right) \beta e^{-\beta b}\mathrm{d}b\\
& =  \int_0^\infty(1 - e^{-\alpha b}) \beta e^{-\beta b}\mathrm{d}b\\
& = \int_0^\infty \beta e^{-\beta b}\mathrm{d}b- \int_0^\infty \beta e^{-(\alpha + \beta)b}\mathrm{d}b\\
& = 1 - \left(\frac{\beta}{\alpha + \beta}\right)\\
& = \frac{\alpha}{\alpha+\beta}
\end{split}
\end{equation*}

\subsection{Multiple Independent Random Variables}
Let $A \sim \exp(\alpha)$, $B_1 \sim \exp(\beta_1)$, $B_2 \sim \exp(\beta_2)$, \ldots $B_m \sim \exp(\beta_m)$ be independent exponentially distributed random variables. Let $C = \min\{B_1, B_2, \ldots, B_m)$, from properties of the exponential distribution, $C$ is also an exponentially distributed random variable with rate $\beta_1+ \beta_2 + \ldots + \beta_m$.

Therefore, 
\begin{equation*}
\begin{split}
P(A \le C) &= P(A - C \le 0) \\ 
& = 1 - \left(\frac{\beta_1+ \beta_2 + \ldots + \beta_m}{\alpha + \beta_1+ \beta_2 + \ldots + \beta_m}\right)\\
&= \frac{\alpha}{\alpha + \beta_1+ \beta_2 + \ldots + \beta_m}.
\end{split}
\end{equation*}

The proof follows by induction on the number of independent exponentially distributed random variables.

\end{document}